\documentclass[pra,twocolumn,showpacs,groupedaddress,notitlepage,superscriptaddress,floatfix,nofootinbib,longbibliography]{revtex4-1}

\usepackage{soul}

\usepackage{array}
\usepackage{float} 
\usepackage{rotating}

\renewcommand{\ne}{ {n^{}_{\rm E}} }

\newcommand{\nb}{ {n^{}_B} }

\newcommand{\ns}{{n^{}_S}}

\newcommand{\na}{{d^{}_A}}

\newcommand{\Is}{I_S}

\newcommand{\Ib}{I_B}

\newcommand{\Ub}{\bar{U}}

\newcommand{\btab}{\begin{tabular}}
\newcommand{\etab}{\end{tabular}}

\newcommand{\ket}[1]{\mbf{|}#1\mbf{\rangle}}
\newcommand{\bra}[1]{\mbf{\langle}#1\mbf{|}}
\newcommand{\braket}[2]{\mbf{\langle}#1\mbf{|}#2\mbf{\rangle}}

\newcommand{\normtwo}[1]{\norm{#1}_{\rm 2}}

\newcommand{\rank}{{\bf rank}}

\newcommand{\trace}{{\bf Tr}}
\newcommand{\avg}{{\bf E}}

\newcommand{\diag}{{\rm diag}}

\newcommand{\eps}{\varepsilon}

\newcommand{\gam}{\gamma}
\newcommand{\alf}{\alpha}
\newcommand{\om}{\omega}
\newcommand{\lam}{\lambda}
\newcommand{\bet}{\beta}
\newcommand{\del}{\delta}
\newcommand{\sig}{\sigma}

\newcommand{\Del}{\Delta}
\newcommand{\Gam}{\Gamma}
\newcommand{\Om}{\Omega}

\newcommand{\norm}[1]{ \left\| #1 \right\| }

\newcommand{\abs}[1]{\left|#1\right|}

\newcommand{\Lcal}{ {\mathcal L} }

\newcommand{\eg}{\emph{e.g.}}
\newcommand{\ie}{\emph{i.e.}}
\newcommand{\etc}{\emph{etc.}}

\newcommand{\bquem}{\begin{quote}\begin{em}}
\newcommand{\equem}{\end{em}\end{quote}}

\newcommand{\blist}{\begin{description}}
\newcommand{\elist}{\end{description}}

\newcommand{\bquote}{\begin{quote}}
\newcommand{\equote}{\end{quote}}

\newcommand{\ben}{\begin{enumerate}}
\newcommand{\een}{\end{enumerate}}

\newcommand{\bit}{\begin{itemize}}
\newcommand{\eit}{\end{itemize}}

\newcommand{\bea}{\begin{array}}
\newcommand{\eea}{\end{array}}

\newcommand{\bds}{\begin{displaystyle}}
\newcommand{\eds}{\end{displaystyle}}

\newcommand{\Rbf}{{\mathbf R}}

\newcommand{\ds}{\displaystyle}

\newcommand{\mbf}[1]{\mbox{\boldmath $#1$}}

\newcommand{\refsec}[1]{\ref{sec:#1}}

\newcommand{\mathbox}[1]{
\fbox{$\ds #1 $}
}

\makeatletter
\def\beq{\@ifnextchar 
[{\@tempswatrue\@beq}{\@tempswafalse\@beq[]}}
\def\@beq[#1]{\begin{equation}\edef\@tmparg{#1}\ifx\@tmparg\@e
mpty \else
	\label{#1}\fi}
\newcommand{\eeq}{\end{equation}}
\makeatother 

\newcommand{\beqaa}{\begin{eqnarray*}}
\newcommand{\eeqaa}{\end{eqnarray*}}

\newcommand{\beqa}{\begin{eqnarray}}
\newcommand{\eeqa}{\end{eqnarray}}

\newcommand{\bc}{\begin{center}}
\newcommand{\ec}{\end{center}}

\newcommand{\bfig}{\begin{figure}}
\newcommand{\efig}{\end{figure}}

\renewcommand{\normtwo}[1]{\norm{#1}_{\elltwo}}

\renewcommand{\Is}{I_n}
\renewcommand{\ns}{{n}}

\newcommand{\nq}{\nc}

\renewcommand{\na}{{n^{}_A}}

\renewcommand{\rank}{\mbox{\rm rank}}

\usepackage[colorlinks=true,
            urlcolor=blue,
            citecolor=red,
            pdfstartview=FitH,
            pdfpagemode=None]{hyperref}

\newcommand{\red}[1]{\textcolor{red}{#1}}

\usepackage{graphicx}
\usepackage[dvips]{epsfig}

\usepackage{times,amsmath,amsthm,amssymb,amsfonts,tikz}

\newtheorem{thm}{Theorem}
\newtheorem{lem}{Lemma}

\usepackage{bbm}

\renewcommand{\trace}{{\rm Tr}}
\newcommand{\tf}{{T}}
\newcommand{\Vcal}{\mathcal{V}}

\newcommand{\Mcal}{{\mathcal M}}
\newcommand{\Xcal}{{\mathcal X}}
\newcommand{\Fcal}{{\mathcal F}}

\newcommand{\Navg}{{N_{\rm avg}}}

\newcommand{\Npwc}{{N_{\rm pwc}}}

\newcommand{\Acal}{{\cal A}}

\newcommand{\mat}[1]{\left[\begin{matrix}#1\end{matrix}\right]}

\newcommand{\beasep}[1]{\renewcommand{\arraystretch}{#1}\bea}

\newcommand{\n}{d}
\renewcommand{\nb}{{d^{}_B}}
\newcommand{\qb}{{q_B}}
\renewcommand{\ns}{{d^{}_S}}
\renewcommand{\na}{{d^{}_A}}
\renewcommand{\nq}{{d^{}_Q}}

\renewcommand{\ne}{{d^{}_E}}

\newcommand{\nucnorm}[1]{\norm{#1}_{\rm nuc}}

\newcommand{\Ws}{{W_S}}

\renewcommand{\Rbf}{\mathbbm{R}}
\renewcommand{\rank}{\mbox{\bf rank}}

\newcommand{\psin}{\psi_{\rm in}}

\newcommand{\psiin}{\ket{\psi_{\rm in}}}

\newcommand{\psiinc}{\bra{\psi_{\rm in}}}

\newcommand{\Fuj}{F}

\newcommand{\Fwc}{F_{\rm wc}}
\newcommand{\Fave}{F_{\rm avg}}
\newcommand{\Favelow}{F_{\rm avg}^{\rm low}}

\usepackage{xspace}

\newcommand{\vlongrightarrow}[1]{\xrightarrow{\hspace*{#1}}}

\newcommand{\tabref}[1]{{\bf Table}~\ref{#1}\xspace}

\renewcommand{\exp}[1]{{\rm exp}\left\{#1\right\}}

\newcommand{\Fnuc}{F_{\rm nuc}}

\newcommand{\PhiB}{\Phi_B}

\newcommand{\Fnom}{\Fuj_{\rm nom}}

\newcommand{\Jrbst}{J_{\rm rbst}}
\newcommand{\Flb}{F_{\rm lb}}
\newcommand{\Flbw}{F_{\rm bnd,w}}
\newcommand{\Flba}{F_{\rm bnd,a}}
\newcommand{\Fwclow}{\Fuj_{\rm wc}^{\rm low}}
\newcommand{\Omavg}{\Om_{\rm avg}}
\newcommand{\Ommax}{\Omunc}
\newcommand{\Omunc}{\Om_{\rm unc}}
\newcommand{\Omeff}{\Om_{\rm bnd}}
\newcommand{\dOmavg}{\Om_{\rm avg}^{\rm dev}}

\renewcommand{\avg}[1]{{\big{\langle}} #1 {\big\rangle}}
\usepackage{bm}
\DeclareBoldMathCommand\boldlangle{\left\langle}
\DeclareBoldMathCommand\boldrangle{\right\rangle}

\newcommand{\Hunc}{{H_{\rm unc}}}
\newcommand{\Hcalunc}{{{\cal H}_{\rm unc}}}

\newcommand{\normsm}[1]{\left\Vert #1 \right\Vert}

\renewcommand{\Is}{I_{S}}
\newcommand{\Ia}{I_{A}}

\newcommand{\Hi}{{\widetilde H}}
\newcommand{\Hiunc}{\Hi}
\newcommand{\Ui}{{\widetilde U}}
\newcommand{\dotUi}{\dot{\Ui}}

\newcommand{\Hm}{H_{\cal M}}

\newcommand{\Hs}{{H_{S}}}

\newcommand{\Ha}{{H_{A}}}
\newcommand{\Hscoh}{{H_{S}^{\rm coh}}}
\newcommand{\Hacoh}{{\Delta_{A}^{\rm coh}}}

\newcommand{\Gsunc}{\widetilde{H}^{\rm coh}_S}
\newcommand{\Gsalf}{\widetilde{S}_{\alpha}}
\newcommand{\Gbalf}{\widetilde{B}_{\alpha}}

\newcommand{\Balf}{{B_\alf}}
\newcommand{\Salf}{{S_\alf}}
\newcommand{\Gsb}{\Hi_{SB}}
\newcommand{\Hb}{H_{B}}
\newcommand{\Hsb}{H_{SB}}
\newcommand{\Us}{U_{S}}

\renewcommand{\Ub}{U_{B}}
\newcommand{\sumalf}{\sum_{\alf}}

\renewcommand{\normtwo}[1]{\norm{#1}}

\newcommand{\normfro}[1]{\norm{#1}_{\mathrm{F}}}

\usepackage[capitalise]{cleveref}

\newcommand{\RLK}[1]{\red{[RLK: #1]}}

\def\>{\rangle}
\def\<{\langle}
\def\Tr{\mathrm{Tr}}

\newcommand{\ketbra}[1]{|{#1}\>\!\<#1|}

\newcommand{\bk}[2]{\<{#1}|{#2}|{#1}\>}

\newcommand{\ketb}[2]{|{#1}\>\!\<#2|}

\usepackage{esvect}
\usepackage[space]{grffile}

\begin{document}

\title{A Fundamental Bound for Robust Quantum Gate Control}

\author{Robert L. Kosut}
\affiliation{SC Solutions, San Jose, CA}
\affiliation{Princeton University, Princeton, NJ 08544, USA}
\affiliation{Quantum Elements, Inc., Thousand Oaks, California, USA}
\author{Daniel A. Lidar}
\affiliation{Quantum Elements, Inc., Thousand Oaks, California, USA}
\affiliation{Departments of Electrical \& Computer  Engineering, Chemistry, and  Physics \& Astronomy, Center for Quantum Information Science \& Technology, University of Southern California, Los Angeles, California 90089, USA}
\author{Herschel Rabitz}
\affiliation{Princeton University, Princeton, NJ 08544, USA}

\begin{abstract}
We derive a universal performance limit for coherent quantum control in the presence of modeled and unmodeled uncertainties.
For any target unitary $W$ that is implementable in the absence of error, we prove that the worst-case (and hence the average) gate fidelity obeys the lower bound
$F \ge \Flb\bigl(\tf \Omeff\bigr)$,
where $\tf$ is the gate duration and $\Omeff$ is a single frequency-like measure that aggregates \emph{all} bounded uncertainty sources, e.g., coherent control imperfections, unknown couplings, and residual environment interactions, without assuming an initially factorizable system-bath state or a completely positive map.
The bound is obtained by combining an interaction-picture averaging method with a Bellman-Gronwall inequality and holds for any finite-norm Hamiltonian decomposition. Hence it applies equally to qubits, multi-level qudits, and ancilla-assisted operations.
Because $\Flb$ depends only on the dimensionless product $\tf\Omeff$, it yields a device-independent metric that certifies whether a given hardware platform can, in principle, reach a specified fault-tolerance threshold, and also sets a quantitative target for robust-control synthesis and system identification.

We translate the theory into a two-objective optimization problem that minimizes both the nominal infidelity and the time-averaged error generator. As an illustrative example we consider a single-qubit Hadamard gate subject to an unknown $\sigma_z$ system-bath coupling; we obtain a five-pulse piecewise-constant control achieving a nominal error of $10^{-7}$ while virtually nulling the average disturbance.
Monte Carlo simulations confirm that every observed infidelity lies below the predicted $\Flb$ curve and that the bound is tight to within one order of magnitude in the relevant regime $1-F\lesssim10^{-4}$.
Our results provide a falsifiable benchmark for experimental characterization as well as a pathway toward error budgets compatible with scalable quantum information processing.
\end{abstract}


\maketitle

\section{Introduction}
\label{sec:intro}

Quantum processors have progressed well beyond laboratory
proofs of concept, yet they remain far from the fully
fault-tolerant regime envisioned for large-scale
computation~\cite{nisq:2018}.  Current resource estimates indicate that
the physical-to-logical qubit ratio required for fault
tolerance is still prohibitive
\cite{Gottesman:2022,Microsoft:2022,Preskill-Megaquop}.  In most
architectures the dominant cost driver is the physical two-qubit gate
error rate; reducing infidelities would translate directly into a corresponding reduction of overhead, although the exact savings depend on device specifics,
error correcting code, and layout
constraints~\cite{Xu:2024aa,Acharya:2025aa}.

A substantial part of this overhead can
be avoided by \emph{maximizing robustness} to all disturbances that
ultimately trigger error correction
\cite{Lidar-Brun:book}.  As Feynman presciently warned,
uncontrolled interactions ``may produce considerable havoc''
in a quantum computer~\cite{Feynman:85}.  If those interactions are
suppressed \textit{before} allocating error-correction resources, the
number of ancilla qubits, circuit depth, and other costs can be significantly
reduced.

There are generally two paths to potentially achieve small
  infidelities in the laboratory setting with qubits: (1) Starting with a
  model of the system and environment, achieve a control design that is robust to
  simulated conceivable uncertainties for transfer to the laboratory for
  performance evaluation. (2) Start directly in the laboratory,
  likely guided by (1), physical motivation, and insights. 
  Due to a host of uncertainties being present, the collective literature shows
  that neither of these approaches have proved to be fully
  satisfactory, especially for two qubit gates. 
  This paper takes neither of these
  approaches, but rather introduces a new theoretical framework and
  associated mathematical analysis.
Our method builds on a cornerstone of classical robust
control---\textit{uncertainty modeling}---in which disturbances are
treated as ``unknown but bounded'' elements of a well-defined set
\cite{Zames:1966,DesVid:1975,LMI94,ZhouDG:96}.  This naturally raises a
fundamental question: \textit{given} such a model, \textit{what is the
ultimate performance limit} of any control strategy?  
Here, we lay the analytical groundwork for answering that question.

\begin{figure}[t]
  \centering
  \includegraphics[width=0.5\textwidth]{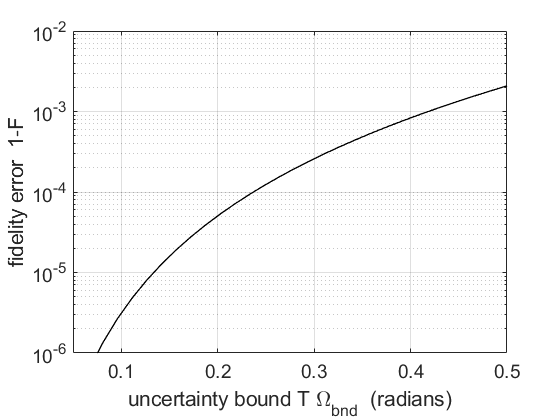}
  \caption{Plot of infidelity upper bound $1-\Flb$ from \cref{eq:Fbnd} in Theorem~\ref{avgThm}. The bound is shown in a log scale 
  vs. uncertainty $\tf\Omeff$ in radians.} 
  \label{lim}
\end{figure}

Our main theoretical result establishes an explicit upper bound on worst-case
infidelity as a function of a single, dimensionless
time-bandwidth uncertainty quantity $\tf\Omeff$ (see Theorem~\ref{avgThm} in
\cref{sec:main}).  \Cref{lim} plots this bound. Here $\tf$ is the gate time and
$\Omeff$ is an aggregate frequency that upper-bounds all relevant
terms in the system and system-bath Hamiltonians
[\cref{eq:Omunc def}]. 

Theorem~\ref{avgThm} follows from the
classical Method of Averaging~\cite{Hale:80} and a specialized
Bellman-Gronwall inequality~\cite{CoddLev:55}.  Although the bound is
not guaranteed to be tight, it delivers a quantitative measure for
both analysis and synthesis: any device that can implement the target
gate perfectly in the uncertainty-free model must, in the presence of
bounded uncertainty, achieve an actual fidelity no worse than
$\Flb(\tf\Omeff)$.  Conversely, control and design choices that lower
$\tf\Omeff$ automatically tighten the bound.

The bound provides a metric to compare with the
  results of either method (1) and/or (2). If the observed
  infidelity is too big, that outcome carries the message
  that additional relevant details need to be included in the model
  for case (1) and/or for case (2) improvements need to be made in the
  platform design and possibly its operational features. In contrast,
  if the observed infidelity is well below the bound, the situation is of
  course attractive; this circumstance could serve as a basis to
  stimulate an interchange between the laboratory and modeling
  efforts to improve the bound, possibly providing for a new and even
  better infidelity to be reached in the future.

Two features distinguish our approach from traditional fidelity
estimates.  First, the bound is derived for \textit{arbitrary} initial
states, including maximally entangled system-bath inputs, and therefore
does not rely on the completely-positive (CP) map framework that
follows from a product state
assumption~\cite{NielsenC:00}. This is essential for realistic
circuit execution, where the system becomes entangled with its
environment between error-correction
cycles~\cite{Schindler27052011}, thus precluding a CP map treatment \cite{Jordan:2004aa,Carteret:05,Rodriguez:08,Buscemi:2013,Dominy:14,Dominy:2016xy}.
Second, $\tf\Omeff$ captures
\textit{all} error sources: the ``known unknowns'' included in the
model and the ``unknown unknowns'' that inevitably remain.  A design
based on the former might, for example, predict $\tf\Omeff\le 0.1$
rad, certifying an infidelity below $10^{-5}$ in
\Cref{lim}.  Even if hidden errors double the bound to $0.2$
rad, the guaranteed error still stays beneath $10^{-4}$.  Numerical
evidence in \cref{sec:numex} shows that the Hamiltonian control acts in real time to steer the system while additionally providing robustness akin to the effect of feedback,
in this case, without measurement.

However, there is a limit to the types of unknown errors that can be included in the time-bandwidth uncertainty $\tf\Omeff$. For example, catastrophic errors such as qubit loss
are quite common in atomic systems:
an atom can fly out of the optical lattice; more common are erasure errors where the state of a qubit is completely reset to the ground state or the maximally mixed state. 
Such catastrophic errors can only be dealt with by error correction~\cite{Wu:2022aa}, not by robust control.

In summary, to apply the theory to a specific device in order to determine the performance limit as \cref{lim} indicates,
it is first necessary to determine the level of uncertainty $\tf\Omeff$.
As physics knowledge alone may not be sufficient,
there is a need to develop a complementary
data-driven uncertainty estimation procedure similar to those that have been developed for classical robust control \cite{id4c:92,KosutLB:92,modval:94}.  
Experiments such as those in \cite{tripathi2024-DB} for parameter estimation would need to be modified for uncertainty estimation, assisting and/or bypassing the
need for a detailed microscopic model, thereby providing the
information required by a compatible robust design framework.

The remainder of the paper formalizes the above ideas and demonstrates
their practical relevance through a detailed numerical example. Uncertainty modeling and fidelity measures are defined
in \cref{sec:model}-\refsec{fid}. The main theoretical framework
establishing a limit of robust performance is presented in
\cref{sec:main}, interpretations are given in \cref{sec:interp},
forms of robust optimization are discussed in \cref{sec:rbst opt}, 
an illustrative numerical example is in \cref{sec:numex},
and \cref{sec:conclude} has concluding remarks. Proofs
are deferred to the Appendices, along with a
sketch of various extensions of the framework.

Before describing our theoretical framework, we note that robust quantum control has a rich literature, including dynamical
decoupling
\cite{Viola:02,Santos:2006:150501,Quiroz:2013fv,Kabytayev:2014aa,Genov:2017aa},
geometric  
methods
\cite{spacecurvePRA:2023,spacecurve:2025,automatedgeometricspacecurve:2025},
and many pulse-shaping and optimization strategies tailored to
specific uncertainty classes
\cite{GreenETAL:2013,KosutGB:13,Soare2014,Kaby:2014,Paz-Silva2014,JonckSL:17,Ball2021,HaasHamEngr:2019,Chalermpusitarak2021,Cerfontaine2021,KosutBR:2022,BerberichKosutSchulte:2024,WeidnerETAL:2025,ChenCory:2025}.
Earlier fidelity bounds of a similar character were obtained in
\cite{PhysRevA.78.012308}, but those results neither incorporate the
set-membership uncertainty model nor attain the tightness achieved
here.

\section{Uncertainty Modeling}
\label{sec:model}

\subsection{Errors}

Errors affecting performance can
occur during state preparation, state evolution, and measurement.
Errors in state preparation and measurement (referred to as SPAM) will
certainly corrupt any evaluation of the state evolution even if the
latter is ideal.  These three operations all require control with
differing goals. \emph{We take the view that these are distinct design
  problems.}  As a result we focus on making the state evolution at
the final time as close as possible, despite uncertainties, to a
desired unitary. 
One consequence of this view is a fidelity measure that is strictly a
function of the state evolution over the gate time, and thus is
separated from any issues involved in state preparation or
measurement.

\subsection{Open bipartite system}

To illustrate the main ideas, we focus on a quantum gate
represented by an open bipartite system evolving over
the finite time interval $t\in[0,\tf]$.
The block diagram below shows an input/output representation of a
system $S$ coupled to a bath $B$.
\beq[eq:SBsys]
\bea{ll}  
\bea{rcl}
\bea{cc}
S & \vlongrightarrow{0.5in}
\\
& \ket{\psi(0)}
\\
B & \vlongrightarrow{0.5in}
\eea
&
\hspace{-1.6ex}
\!\!\mathbox{
  \bea{ccc}
  &&\\&U(t)&\\&&
  \eea
}
&
\hspace{-1.4ex}
\bea{cc}
\!\!\!\vlongrightarrow{0.5in} & {S}
\\
\ket{\psi(t)} & 
\\
\!\!\!\vlongrightarrow{0.5in} & {B}
\eea
\eea
\eea
  \eeq
The system and bath Hilbert space dimensions are $\ns$ and $\nb$, respectively,
with finite bath dimension $\nb$, though possibly large. This gives a finite total dimension
of $\n=\ns\nb$. The corresponding $\n$-dimensional unitary evolution $U(t)$ and state 
$\ket{\psi(t)}$ are given by
\beq[eq:Uevo]
\beasep{1.25}{lll}
  \dot U(t) &=& -i H(t)U(t),\quad U(0)=I
  \\
  \ket{\psi(t)} &=& U(t)\ket{\psi(0)},
  \quad
  \ket{\psi(0)}=\psiin
  \eea
\eeq
Here $\hbar=1$, hence the total system-bath Hamiltonian $H(t)$ is in units of
radians/sec or $H(t)/2\pi$ in Hz.


\subsection{Modeling assumptions}
As indicated in \cref{eq:Uevo}, we assume that the initial system-bath state is a pure 
$\n$-dimension state $\psiin \equiv \ket{\psi(0)}$, but not necessarily a product state. This will allow us to account for system-bath entanglement due to a prior gate operation. 
The corresponding $\n$-dimensional bipartite system Hamiltonian $H(t)$ is given by, 
\beq[eq:Ham]
H(t) = \left(\Hs(t)+\Hscoh(t)\right)\otimes I_B 
+ I_S\otimes H_B + \Hsb
\eeq
where $\Hs(t)$ is an assumed model of the uncertainty-free system
with the uncertainty-free unitary $\Us(t)$ obtained from,
\beq[eq:Usnom]
\dot{U}_S(t) = -i\Hs(t)\Us(t),\ \Us(0)=I_S
\eeq
The uncertainty-free system Hamiltonian can often be arranged to be of the form, 
\beq[eq:Hsnom]
\Hs(t) = {\Hs}_0+
\sum_j v_j(t){\Hs}_j
\eeq
with control variables 
$v_j(t)\in\Rbf,t\in[0,\tf]$.

The uncertain parts of the Hamiltonian \cref{eq:Ham} are the \emph{coherent error} $\Hscoh(t)$,  
the \emph{bath self-dynamics} $\Hb$,  and the \emph{system-bath coupling} $\Hsb$. 
The bath Hamiltonians $\Hb,\Hsb$ are  assumed constant but uncertain 
during any gate time operation $t\in[0,\tf]$. 
Thus, we define
\beq[eq:Hunc]
\Hunc(t) = \Hscoh(t)\otimes\Ib + \Is\otimes\Hb + \Hsb
\eeq
as the component of the total Hamiltonian that captures all the uncertainty.

Coherent errors in the system may contain biases and scale factors, some arising from
the signal generator and connectors to the quantum device; thus $\Hscoh(t)$
may depend on the controls. 
The uncertain bath self-dynamics is independent of the uncertainty-free system evolution and obeys
\beq[eq:Ub]
\dot{U}_B(t) = -i\Hb\Ub(t),\ \Ub(0)=\Ib
\eeq

\emph{Decoherence} is due entirely due to the presence of the system-bath 
coupling $\Hsb$, which has the general form,
\beq[eq:Hsbalf]
\Hsb = \sumalf  \Salf\otimes \Balf
\eeq
where $\alf$ denotes the specific coupling mechanism,
\eg, usually $\Salf$ consists of combinations of the Pauli operators $\sig_x,\sig_y,\sig_z$ acting on different qubits.
Obviously, if $\Hsb=0$ then the system and the bath each evolve independently; this is merely sufficient, and in general symmetries in $\Hsb$ give rise to noiseless subsystems wherein the system dynamics are purely unitary~\cite{Zanardi:97c,Knill:2000dq,ShabaniLidar:05a}.


\section{Fidelity}
\label{sec:fid}

\subsection{Uhlmann fidelity} 

The Uhlmann fidelity between two states $\rho$ and $\sigma$ is \cite{Uhlmann},
\beq[eq:Uhl fid]
\Fcal(\rho,\sigma)= \trace\sqrt{\sqrt{\rho}~\sigma\sqrt{\rho}}
\eeq 
When $\sigma$ is a pure state 
$\ketb{\psi}{\psi}$ this reduces to $\Fcal(\rho,\psi)=\sqrt{\bra{\psi}\rho\ket{\psi}}$, and when also $\rho$ is a pure state $\ketb{\phi}{\phi}$, we have 
$\Fcal(\phi,\psi)=|\braket{\psi}{\phi}|$.\footnote{Fidelity is sometimes defined as the square of \cref{eq:Uhl fid}, \eg, \cite{NielsenC:00} vs. \cite{Jozsa:94,Gilchrist:2005}.}
For the bipartite system \cref{eq:SBsys}, assuming a decoupled initial state 
$\psiin = \ket{\psi_S}\otimes\ket{\psi_B}$, 
the map from the $S$-channel input density matrix $\rho_{\rm in}=\ket{\psi_S}\bra{\psi_S}$ 
to the $S$-channel output density matrix $\rho_S=\Tr_B[U(T)\rho_{\rm in}U^\dagger(T)]$ 
is \emph{completely positive and trace preserving} (CPTP). 
However, as already noted, since consecutive inputs to gates are unlikely to be decoupled from the bath,
a CPTP map is not an accurate model for our purposes \cite{Jordan:2004aa,Rodriguez:08,Dominy:2016xy}.
Moreover, the bath coupling errors may accrue over many repetitions, rendering
$\Fcal(\rho_S,\rho_{\rm in})$ as an ineffective measure to evaluate robustness.
Instead, we consider an arbitrary pure system-bath state as the input to any gate operation. Rather than tracing out the bath and computing the fidelity between the desired and actual \emph{reduced} system states, we do so with the complete system-bath state.

\subsection{Design goal}

Referring to \cref{eq:Uevo}, for any pure input system-bath state $\psiin$, 
the final-time output state is,
\beq[eq:psiout]
\ket{\psi(\tf)}=U(\tf)\psiin
\eeq
The ideal design goal is that the final-time unitary $U(\tf)$ factors into a tensor product over $S$ and $B$. Thus
the ideal desired output state at the final-time is,
\beq[eq:psides ideal]
\ket{\psi_{\rm des}} = (\Ws \otimes W_B)\psiin
\eeq
where $\Ws$ is the $\ns\times\ns$ target unitary for the system
channel and where $W_B$ is any $\nb\times\nb$ bath unitary at the final time.
In Appendix \refsec{Fnuc}, following \cite{KosutGBR:06,GDKBH:10}, we outline how $W_B$ can be used as a free design variable to improve performance.
For the present analysis, it suffices 
to select $W_B=\Ub(\tf)$, the specific final-time bath unitary 
evolving from \cref{eq:Ub}. 
Thus the desired output state at the final-time is,
\beq[eq:psides]
\ket{\psi_{\rm des}} = (\Ws \otimes \Ub(\tf))\psiin
\eeq

\subsection{Fidelities}
\paragraph*{State fidelity} The fidelity between the final-time output state \cref{eq:psiout} and the desired state \cref{eq:psides} is
\beq[eq:Fstate]
\beasep{1.25}{rcl}
F(\psin) &\equiv& F(\psi_{\rm des},\psi(\tf)) = |\braket{\psi_{\rm des}}{\psi(\tf)}|
\\
&=&
|\psiinc\big(\Ws\otimes\Ub(\tf)\big)^\dag U(\tf)\psiin|
\eea
\eeq

\paragraph*{Worst-case fidelity} Defined over all pure input states by,
\beq[eq:Fwc]
\Fwc \equiv \min_{\psin}F(\psin)
\eeq

\paragraph*{Average fidelity}
Defined over the Haar measure on pure input states,
\beq[eq:F-ave]
\Fave \equiv \int F(\psin) d\psin
\eeq

\paragraph*{Nominal fidelity} Defined as the standard \emph{overlap
  fidelity} \cite{Gilchrist:2005,GDKBH:10} between the nominal
(uncertainty-free) unitary $\Us(t)$ at the final time and the target
unitary:
\beq[eq:Fuj nom]
\Fnom \equiv |\trace(\Ws^\dag \Us(\tf)/\ns)|
\eeq
Note that $\Fnom=1$ iff $\Us(\tf)=\phi\Ws$ with global phase
$|\phi|=1$. The worst-case, average,  
and nominal fidelity do not depend on the input state. All these fidleities evaluate
\emph{only} their respective performance to realize a unitary target. 

\subsection{Interaction picture}

To reveal robust performance properties, the system dynamics and
corresponding fidelity measures are better expressed in terms of the
\emph{interaction-picture unitary},\footnote{All interaction-picture operators are denoted by a tilde, e.g., $\Ui,\Hi$.}
\beq[eq:Rint]
\beasep{1.25}{rcl}
\Ui(t) &=& \Big( \Us(t)\otimes \Ub(t) \Big)^\dag U(t)
\eea
\eeq
which evolves as, 
\beq[eq:Rint evo]
\beasep{1.25}{rcl}
\dotUi(t) &=& -i\Hiunc(t)\Ui(t),
\quad
\Ui(0)=I
\eea
\eeq
Using the modeling assumptions from \cref{eq:Ham}-\cref{eq:Hsbalf}, results in
the \emph{interaction-picture uncertainty Hamiltonian} $\Hiunc(t)$ given explicitly by,
\beq[eq:Gint]
\beasep{1.25}{rcl}
\Hiunc(t) &=& \Gsunc(t)\otimes\Ib + \Gsb(t)
\eea
\eeq
with the indicated interaction-picture Hamiltonians,
\beq[eq:Gint1]
\beasep{1.25}{rcl}
\Gsunc(t) &=& \Us(t)^\dag\Hscoh(t)\Us(t)
\\
\Gsb(t) &=& \sumalf \Gsalf(t)\otimes \Gbalf(t)
\\
\Gsalf(t) &=& \Us(t)^\dag \Salf\Us(t)
\\
\Gbalf(t) &=& \Ub(t)^\dag \Balf \Ub(t)
\eea
\eeq

\subsection{Fidelity via interaction-picture unitary}
In terms of the final-time interaction-picture unitary $\Ui(\tf)$ defined in
\cref{eq:Rint}, the input-state dependent fidelity \cref{eq:Fstate}, 
now becomes,
\beq[eq:FstateR]
F(\psin) = |\psiinc(\Ws^\dag \Us(\tf)\otimes I_B)\Ui(\tf)\psiin|
\eeq
while the corresponding worst-case fidelity and average fidelity are still given by \cref{eq:Fwc} and \cref{eq:F-ave}, respectively, with $F(\psin)$ as in \cref{eq:FstateR}.
If the target unitary
$\Ws$ is in the reachable set of the uncertainty-free system, then for some $\Hs(t)$ 
the nominal fidelity $\Fnom=1$ in \cref{eq:Fuj nom} and,
\beq[eq:Fstate Ws]
F(\psin) = |\psiinc\Ui(\tf)\psiin|
\eeq
As shown in \cref{sec:Fpsimin},
the following is a prerequisite for the main result.
\beq[eq:Fnom1]
\beasep{1}{c}
\hline
\mbox{\bf Fidelity Lower Bounds}
\\
\hline\\
  \beasep{2}{ll}
  \mbox{\em If}
  &
  \Fnom = 1 \quad (\mbox{iff $\Us(\tf)=\phi\Ws,|\phi|=1$})
  \\
  \mbox{\em Then}
  &
  \left\{
  \beasep{1.25}{rcl}
  \ds \Fwc &=& \min_{\psin} |\psiinc \Ui(\tf)\psiin| \geq \Fwclow
  \\
  \displaystyle \Fave 
  &=& \int |\psiinc \Ui(\tf)\psiin| d\psin \\
&\ge& \abs{\trace~{\Ui(\tf)/d}} \geq \Favelow
 \eea
 \right.
 \\
  \eea
  \eea
   \eeq

\quad\emph{with he fidelity lower bounds,}
\beq[eq:low]
\beasep{1.5}{l}
\Fwclow \equiv  \max\bigl(1-\frac{1}{2}\|\Ui(\tf)-I\|^2,0\bigr) 
\\
\Favelow \equiv  \max\bigl(1-\frac{1}{2{\n}}\|\Ui(\tf)-I\|_{\mathrm{F}}^2,0\bigr)
\eea
\eeq
Here and henceforth, $\norm{\cdot}$ is the induced $2$-norm (the largest singular value) \cite{DesVid:1975} 
and $\norm{\cdot}_{\mathrm{F}}$ is the Frobenius norm (the square-root of the sum-square of singular values).\footnote{$\norm{\cdot}$ is also commonly known as the operator-norm~\cite{Bhatia:book}:
for any matrix $A$, $\norm{A}$
is the maximum singular value, and if $A$ is Hermitian, then $\norm{A}$ equals the maximum absolute value of the
eigenvalues. The Frobenius norm is the square root of the sum of the squares of the singular values: $\norm{A}_{\mathrm{F}} \equiv \sqrt{\Tr A^\dagger A}$, not to be confused with the trace norm or nuclear norm  (the sum of the singular values).}
A standard norm inequality between the Frobenius and induced $2$-norm is
$\|A\|_{\mathrm{F}}\leq\sqrt{\n}\|A\|$ for any operator $A$. As a result,
\beq[eq:ineq2]
\Favelow \geq \Fwclow
\eeq
as expected.
Also shown in \cref{sec:Fpsimin}, if given the final-time
unitaries $\Us(\tf)$ and $\Ui(\tf)$, then $\Fwc$ can be computed to within
any desired precision via an equivalent convex optimization.

\section{Robust Performance Limit}
\label{sec:main}

As \cref{eq:Fnom1} shows, if $\Fnom=1$ and the final-time interaction-picture unitary 
$\Ui(\tf)\approx I$ then both $\Fwc,\Fave\approx 1$.   Our aim is to find a
limit on how closely this goal can be achieved.  
A direct approach to maximize $\Fwc$ for any input state is to maximize the lower bound 
$\Fwclow$ in \cref{eq:low}.
Equivalently posed as an optimization problem, 
\beq[eq:wcopt]
\beasep{1.5}{ll}
\mbox{minimize}
&
{\ds\max_{\Hcalunc}} \|\Ui(\tf)-I\|
\\
\mbox{subject to}
&
\Hiunc(t)\in\Hcalunc,
\quad 
\vec v(t)=\{v_j(t)\} \in \Vcal = \Rbf^{N_c}
\eea
\eeq
with $\Ui(t)$ and $\Hiunc(t)$ from \cref{eq:Rint evo}-\cref{eq:Gint1}
and where $\Hcalunc$ is a set which characterizes the interaction Hamiltonian uncertainty, see, \eg. \cref{eq:Omunc def}.
The $N_c$ optimization variables are the controls $\vec v(t)$ in $\Hs(t)$ from \cref{eq:Hsnom},
with typical constraints in $\Vcal$ on magnitude, bandwidth, \etc\
 While this problem formulation is direct,  the main issue 
is the potentially prohibitive computational cost for a system with a large bath dimension or with connections to other states in the device, \eg, additional system levels and crosstalk. 
An approach to robust design is described next which 
deals with the computational issues and leads naturally to the main result
as depicted in \cref{lim}.

\subsection{Uncertainty characterization}

We address all these issues 
by first directly bounding infidelity as a function
of specific bounds on components of the uncertain interaction-picture Hamiltonian $\Hiunc(t)$ and its
time-average. 
For the Hamiltonians in \cref{eq:Gint1}, 
and with the time-average for any matrix $A$ defined by,
\beq[eq:avg mag]
\avg{A}=(1/\tf)\int_0^\tf A(t)dt
\eeq
define the following uncertainty bounds:
\beq[eq:Omunc def]
\beasep{1.5}{lcl}
\Ommax &\geq& {\ds\max_t}
\norm{\Hscoh(t)}
+\sumalf\norm{\Salf}~\norm{\Balf}
\\&\geq& {\ds\max_t}\normsm{\Hunc(t)}
\\
\Omavg &\geq&
\normsm{\avg{\Gsunc}}
+\sumalf\normsm{\avg{\Gsalf\otimes\Gbalf}}
\\&\geq& \|\avg{\Hi}\|
\\
\dOmavg &\geq& {\ds\max_t}\normsm{\Hiunc(t)-\avg{\Hiunc}}
\eea
\eeq

Given our earlier choice of setting $\hbar=1$, 
all these measures in \cref{eq:Omunc def} are in
units of frequency, specifically radians/sec, or in Hz when divided by $2\pi$.

The frequency $\Ommax$ reflects mostly intrinsic system errors, whereas
$\Omavg$ and $\dOmavg$ are composed of errors that can be affected by the control
dependent uncertainty-free unitary evolution $\Us(t)$, \ie,
$\Gsunc(t)$ and $\Gsb(t)$ as defined in \cref{eq:Gint1}.
Bounds similar to those in \cref{eq:Omunc def} are common to control
protocols based on dynamical decoupling~\cite{Viola:98,ViolaKL:99,Uhrig:07,KhodjastehLidar:07,DDPRA:2011,XiaGotzLidar:2011}.

\emph{It is important to note that in certain important cases of interest, 
such as bosonic baths, for some Hamiltonian terms the norms in \cref{eq:Omunc def} diverge.} This necessitates replacing the aforementioned norm with a different measure of uncertainty, \eg, one that is input-state dependent, 
such as the correlation functions in \cite{DDPRA:2011}. We defer a treatment along those lines to a future publication, but note that correlation functions are already subsumed in a Lindblad master equation as briefly described in \cref{sec:extend}. However, the convergence of the time-dependent perturbation theory underlying quantum master equations is likewise predicated upon finite operator norms~\cite{Mozgunov:2019aa}.  

The robust performance limit bound displayed in \cref{lim} and discussed in the Introduction is based on the
following theorem. 

\begin{thm}{Robust Performance Limit Fidelity Lower Bound}\label{avgThm}

\vspace{0.10in}
  Given the Hamiltonian uncertainty bounds \cref{eq:Omunc def}, define the dimensionless, effective
  time-bandwidth uncertainty bound, or error bound for short,
  \beq[eq:TOmeff]
  \tf\Omeff \equiv \sqrt{(\tf\Ommax)(\tf\dOmavg) + 4\tf\Omavg}
  \eeq
  with associated fidelity lower bound,
  \beq[eq:Fbnd]
  \Flb =  \max\left(1-\frac{1}{2}\left(e^{(\tf\Omeff/2)^2}-1\right)^2,0\right)
   \eeq
   
   Assume that the nominal fidelity [\cref{eq:Fuj nom}] is maximized,
   that is, $\Fnom=1$, or equivalently, $\Us(\tf)=\phi\Ws$ with global
   phase $|\phi|=1$. 
   
   Then both the worst-case fidelity
   \cref{eq:FstateR} and the average-case fidelity [\cref{eq:FstateR}] are bounded below by $\Flb$, i.e,
   \beq[eq:Fwc Fbnd]
     \beasep{1}{rcl}
   \Fwc &=& {\ds\min_{\psin}} |\psiinc \Ui(\tf)\psiin| \geq \Flb\\
   \Fave &=& \int |\psiinc \Ui(\tf)\psiin| d\psin \geq \Flb
   \eea
   \eeq
\end{thm}

Note that without the $\max$,
\beq[eq:TOm rad]
\beasep{1.25}{c}
0\leq F_{\rm bnd}\leq 1
\\
\text{iff}
\\
\tf\Omeff \leq 2\sqrt{\ln(1+\sqrt{2})}
= 1.8776 \mbox{ radians}
  \eea
\eeq 
which defines a physical range of error bound values for which $\Flb$ provides a non-trivial bound.

\subsection{Sketch of proof}
The full proof in \cref{sec:avgThm} is based on a modified version of the standard transformation of variables used in the classic \emph{Method of Averaging} \cite{Hale:80}. In this case, the variable to be transformed is the interaction-picture unitary. The resulting differential equation highlights the terms that involve time-averaging. When substituted into the norm of the transformed interaction-picture unitary error in \cref{eq:Fnom1}, a bound can be obtained using the terms in \cref{eq:Omunc def} by appealing to a particularly applicable version of the \emph{Bellman-Gronwall Lemma} \cite{CoddLev:55}.

%

\section{Interpretations}
\label{sec:interp}

As previously presented in the Introduction, \cref{lim} shows a 
plot on a logarithmic scale of
the infidelity upper bound $1-\Flb$ versus the effective time-bandwidth
uncertainty bound $\tf\Omeff$.
To utilize the bounding curve to predict expected performance, 
the range of the effective uncertainty level $\tf\Omeff$ needs to
be determined from the device. 
It is important to emphasize (again) that the effective time-bandwidth
uncertainty parameter $\tf\Omeff$  includes \emph{all}  Hamiltonian uncertainties, both those that have been called ``known unknowns''
as well as, by implication, ``unknown unknowns.''  
After incorporating an uncertainty model and robust design, to determine an actual 
bound on $\tf\Omeff$ will undoubtedly require data from experiments
in much the same way as
existing approaches to uncertainty estimation are determined for classical systems,
\eg, \cite{id4c:92,KosutLB:92,modval:94}.

\subsection{Ideal minimum uncertainty measure}

It is reasonable to assume that the system is sufficiently well designed so
that the uncertainty-free model system is completely controllable.  
Thus the fidelity of the uncertainty-free \emph{model} 
can achieve the limit of $\Fnom=1$. If, in addition,  all the time-averaged terms 
directly affected by control \cref{eq:Gint} could be annihilated, that is, 
coherent and system-bath coupling errors, equivalently
$\avg{\Hiunc}=0$, then the effective time-bandwidth 
uncertainty is the smallest possible, namely, $\tf\Omeff=\tf\Omunc$, the intrinsic uncertainty \cref{eq:Omunc def}. Any remaining errors can be further minimized by a combination of other design variables. Under these ideal conditions, stated as a corollary
to Theorem~\ref{avgThm}:
\beq[eq:avgThm TOmunc]
\beasep{1}{c}
\hline
\mbox{\bf Minimum Time-Bandwidth Uncertainty}
\\
\hline
\beasep{1.25}{ll}
\mbox{\em If}\!\!
&
\!\!\!\!\left\{
\beasep{1.25}{l}
\Fnom=1 \quad (\mbox{iff $\Us(\tf)=\phi\Ws,|\phi|=1$})
\\
\avg{\Gsunc}=0
\\
\avg{\Gsalf\otimes\Gbalf}=0,\forall\alf
\eea
\right.
\\
\mbox{\em Then}
& \tf\Omeff = 
\tf\Ommax
\\
&\!\!\!\!
= \sumalf\norm{\Salf}~\norm{\tf\Balf}
+
{\ds\max_t}\norm{\tf\Hscoh(t)}
\eea
\eea
\eeq
The idealized assumptions in \cref{eq:avgThm TOmunc} 
reduce the effective time-bandwidth uncertainty to
the minimum intrinsic value of $\tf\Ommax$ as shown above: the sum of the inherent strength of 
the sum of system-bath couplings and coherent errors.

\subsection{Selected gate times}

\cref{lim1} shows limit bounds versus the effective uncertainty
frequency $\Omeff/2\pi$ in MHz, each bound corresponding respectively
to the three selected gate times displayed that are 
typical of superconducting transmon qubits, \eg, \cite{tripathi2024-DB}.
\tabref{lim-tab} shows specific maximum uncertainty
frequencies in Hz ($\Omeff/2\pi$) to achieve infidelities bounded
by $10^{-4}$ and $10^{-5}$, respectively, for the three gate times.
Obviously the same infidelity bounds could be achieved with a
longer gate time and smaller uncertainty.

\begin{figure}[t]
  \centering
  \includegraphics[width=0.5\textwidth]{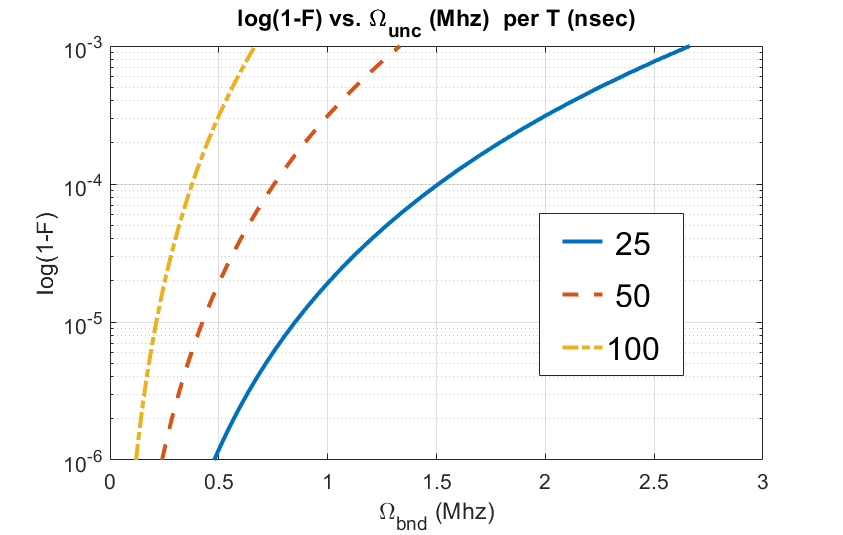}
  \caption{Plot of three performance limit bounds on log of fidelity
    error $1-\Flb$ versus the effective uncertainty $\Omeff$ in MHz for typical gate
    times $\tf\in\{25,50,100\}$ nsec.}
  \label{lim1}
\end{figure}
\begin{table}[t]
  \centering
  \begin{small}
  \btab{|c||c||c||c|}
  \hline
  $1-F\leq$ & $T=25$ ns & $T=50$ ns & $T=100$ ns
  \\
  \hline\hline
  $10^{-4}$ & $1.51$ MHz & $754$ KHz & $377$ KHz
  \\\hline
  $10^{-5}$ & $850$ KHz & $425$ KHz & $213$ KHz
  \\
  \hline
  \etab
  \end{small}
  \caption{Maximum uncertainty frequencies ($\Omeff/2\pi$ in MHz) from
    [Theorem~\ref{avgThm} in \cref{sec:main}] and \cref{lim1} to achieve
    the indicated infidelity bounds on $1-F$ for the three
    selected gate times $\tf$ in nanoseconds (ns).}
  \label{lim-tab}
\end{table}
%

\subsection{Bounding Bath Uncertainty}
\label{sec:bounds}

Maximizing the nominal fidelity while eliminating the time-averaged
coherent interaction term is easily handled by control.  Eliminating,
or greatly reducing, the time-average of the system-bath coupling
terms is more difficult, 
and requires some knowledge of the bath dynamics; with such knowledge, techniques such as dynamical decoupling and quantum error correction can be used toward this end \cite{Lidar-Brun:book}.
%
In addition, without assuming a detailed knowledge of bath dynamics, 
a variety of effective uncertainty bounds $\tf\Omeff$ can be formed 
dependent on assumptions about the bath.
For example, suppose the bath part of the system-bath coupling 
and the bath self-dynamics are both 
approximately known, \ie, 
$\norm{\Balf-\bar{B}_\alpha}\leq\del_B$ and 
$\norm{\Hb-\bar{H}_B}\leq\Del_B$. 
Knowledge of $\bar B,\bar U_B, \del_B, \Del_B$ 
is easily incorporated into the bounds \cref{eq:Omunc def}.
Whatever the assumptions, the resulting effective uncertainty measure 
$\tf\Omeff$ will provide an 
upper bound on predicted infidelity.

\subsection{Unknown unknowns}
Finding controls to ensure that the coherent interaction time-average 
$\normsm{\avg{\Gsunc}}\approx 0$ is very likely.
However, in the face of unknown uncertainties, 
it may not be possible to completely annihilate the time-averaged interaction-picture Hamiltonian
of the system-bath coupling term in \cref{eq:Gint}. When $\normsm{\avg{\Hi_{SB}}}>0$,
it follows that $\Omavg=\normsm{\avg{\Hiunc}} > 0$.
In this case 
the effective time-bandwidth uncertainty bound, $\tf\Omeff$, 
contains \emph{all} uncertainties, both those known and unknown. In the ideal case   
when $\Omavg=0$, an assumption in \cref{eq:avgThm TOmunc}, $\Omeff$ reduces to $\Ommax$. 
The robust performance bound \cref{lim} can be used to give an approximate accounting of the effect of the inevitable unknown uncertainties.

For example, if the designed model based on ``known unknowns'' 
yields $\tf\Omunc\leq 0.15$ radians, then the 
corresponding upper bound on infidelity is $1-\Flb = 1.59\times{10}^{-5}$.
A relative uncertainty increase of 100\% 
from unknown sources to $\tf\Omunc\leq 0.30$ radians yields  $1-\Flb = 2.59\times{10}^{-4}$,
more than a 16-fold increase in infidelity, but still below a $10^{-3}$ error.
As previously stated, 
even if the effective uncertainty increases substantially, that does not mean the infidelity will also.
The bounding curve \cref{lim} thus provides a reasonable assurance that no matter how the system is designed, even in the face of unknown uncertainties unaccounted for in the design model, a small infidelity could still accrue and all may be well. The numerical example in \cref{sec:numex} provides further assurance.

\section{Robust Optimization}
\label{sec:rbst opt}

The main result on the limit of robust performance, Theorem~\ref{avgThm}, provides 
a means, and criteria, for both \emph{analysis}  and
\emph{synthesis} of a robust design for a controlled quantum gate.  
Specifically, to make $\Fnom=1$
the final time nominal system unitary $\Us(\tf)$ should be very close
to the target $\Ws$, and simultaneously, the terms in the
time-bandwidth uncertainty $\tf\Omeff$ which are dependent on its
evolution over $t\in[0,T]$ should be as small as possible. This suggests that a 
robustness measure for optimization is the magnitude of all
time-averages of interaction Hamiltonians dependent on the control
variables that manipulate the evolution of the 
\emph{uncertainty-free system unitary}, $\Us(t),t\in[0,\tf]$.
Symbolically representing the controls by $v$, the optimization measures are,   
\beq[eq:Jrbst set]
\beasep{1.25}{l}
\Fnom(v) = |\trace(\Ws^\dag\Us(\tf)/\ns)|^2
\\
\Jrbst(v) =
\max\left\{
\beasep{1.5}{l}
\norm{\avg{\Us^\dag\Hscoh\Us}},
\\
\norm{
  \avg{
    \Us^\dag\Salf\Us\otimes
    (\Ub^\dag\Balf\Ub)
  }
  },\forall\alf
\eea
\right\}
\eea
\eeq
The $\alf$-dependent terms require a model of the bath. 
With no knowledge of the bath, 
the robutness measure reduces to,
\beq[eq:Jrbst set 1]
\Jrbst(v) =
\max\left\{
\beasep{1.25}{l}
\norm{\avg{\Us^\dag\Hscoh\Us}},
\\
\norm{
  \avg{
    \Us^\dag\Salf\Us  
    }
  },\forall\alf
\eea
\right\}
\eeq
As previously discussed in \cref{sec:interp} there are a variety of possibilities depending on approximate bath modeling assumptions.

There are also constraints on the control variable $v$ that are platform dependent. For example, $v$ may originate 
from a waveform generator that is driven by a command signal $\bar v$. The constraint 
$v\in\Vcal$ characterizes the relationship, \eg, $\Vcal$ delineates the constraints on magnitude, power, bandwidth, sampling rate, \etc\ Such constraints, if not taken into account, can have a significant affect on performance, \eg, \cite{AllenPRA:17}.

Regardless of the form of the robustness measure and control constraint set, simultaneous minimization of the nominal infidelity $1-\Fnom(v)$
and $\Jrbst(v)$ subject to $v\in\Vcal$ has been presented in various
ways in~\cite{Viola:98,ViolaKL:99,Uhrig:07,KhodjastehLidar:07,DDPRA:2011,XiaGotzLidar:2011,GreenETAL:2013,Soare2014,Kaby:2014,Paz-Silva2014,Ball2021,HaasHamEngr:2019,Chalermpusitarak2021,Cerfontaine2021}.

For example, consider a \emph{single-stage optimization},
\beq[eq:onestage]
\beasep{1.25}{ll}
\mbox{minimize}&
1-\Fnom(v) + \lam \Jrbst(v)
\\
\mbox{subject to}& v\in\Vcal
\eea
\eeq
where $\lam$ is a preselected
parameter that weighs the relative objectives.
Alternately, the \emph{two-stage optimization} described in \cite{KosutBR:2022},
first maximizes only the
nominal fidelity $\Fnom(v)$.  When this fidelity crosses a high
threshold, $f_0\approx 1$, the optimization switches to minimizing the robustness
measure $\Jrbst(v)$ while keeping $\Fnom(v)$ above $f_0$.  This
results in the following formulation,
\beq[eq:twostage]
\beasep{1.25}{l}
\mbox{\em Stage 1}\
\max \Fnom(v),\ v\in\Vcal
\\
\mbox{{\em Stage 2} when $\Fnom(v) \geq f_0$}
\\
\quad\quad\quad
\left\{
\beasep{1.25}{ll}
\mbox{minimize} & \Jrbst(v)
\\
\mbox{subject to}&
\Fnom(v) \geq f_0,\ v\in\Vcal
\eea
\right.
\eea
\eeq
No matter the formulation, 
\emph{the quantum control design problem is not a convex optimization.} 
It is a subset of the classical bilinear control problem where the control
multiplies the state. All optimization methods are iterative, 
and there is no one-shot solution except for
some exceptional cases, \eg, \cite{Krener:78}. 
However, the freedom to minimize both the infidelity $1-\Fnom(v)$ 
and robustness measure $\Jrbst(v)$ 
is known to arise from the
ability to roam over the null space at the top of the fidelity
landscape~\cite{RabitzHR:04,Ho.PRA.79.013422.2009,HsiehWRL:10,BeltraniETAL:2011,MooreR:2012,HockerETAL:2014,RussellRW:17,KosutArenzRabitz:2019}. 
The structure of the quantum control landscape, despite being ``bumpy'' with numerous saddles and seldom (topologically ``almost never'') contains local optima, generally leads to convergence. 

\begin{figure*}[t]
\centering
\btab{ccc}
\hspace{-0.2in}
\includegraphics[width=0.33\textwidth]{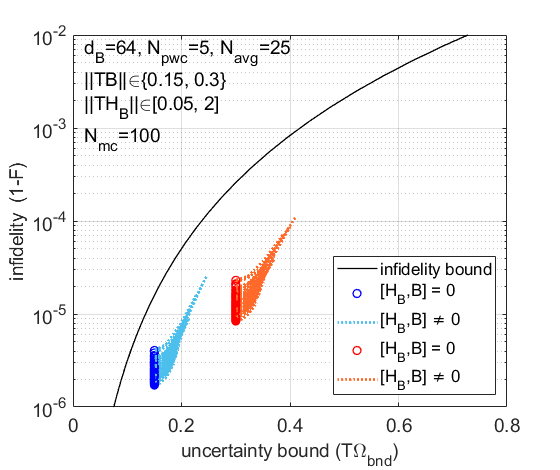}
&
\includegraphics[width=0.33\textwidth]{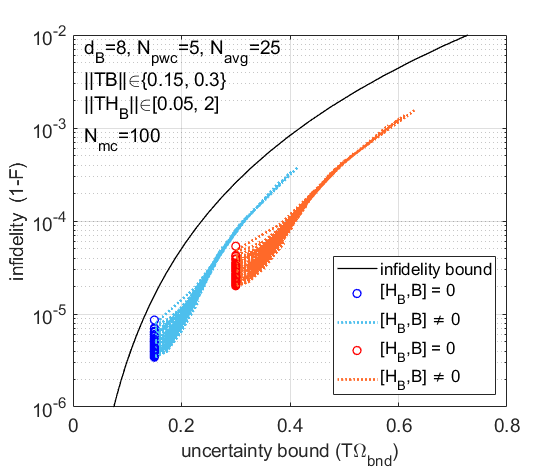}
&
\includegraphics[width=0.33\textwidth]{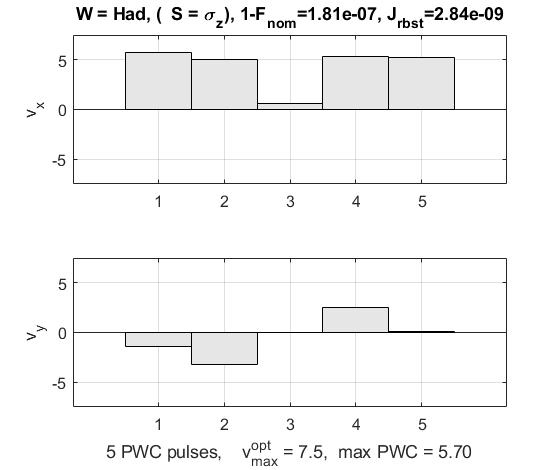}
\etab
\caption{
Results obtained from solving \cref{eq:opt1qu const0} with normalized gate time of $\tf=1$.
{\bf Left $(\nb=64)$ and Middle ($\nb=8$): Infidelity vs. uncertainty} Black line is the limit bound from Theorem~\ref{avgThm}. Blue and red circles are $1-\Fwclow$ from \cref{eq:Fnom1} for bath Hamiltonians $(\Hb,B)$ 
which commute in $\sig_x$ as in \cref{eq:HBB aligned} with the two values for $\|TB\|$ given in \cref{eq:eigs HBB}: red is 0.3, 
blue 0.15. 
The light blue and red points are plots of $1-\Fwclow$ showing 100 samples each of 8 uniformly spaced samples from  the range for 
$\norm{H_{SB}}$ from \cref{eq:eigs HBB} of the non-commuting coefficients $(h,g)$.
{\bf Right: Piece-wise-constant (PWC) pulses} over normalized gate time $t\in[0,1]$ with $\Npwc=5$
and control magnitude constraint set at $v_{\max}=7.5$ in 
\cref{eq:opt1qu const0}. The largest control magnitude achieved is 5.70 in the $\sigma_x$ control. 
}
\label{had}
\end{figure*}

\section{Numerical example}
\label{sec:numex}

\subsection{Single qubit system}
Consider a single qubit system with controls in $\sig_x$ and $\sig_y$,
no coherent errors, and known to be coupled via $\sig_z$ to an uncertain time-independent bath. The resulting model Hamiltonian is, 
\beq[eq:ex const]
\beasep{1.1}{rcl}
H(t) &=& H_S(t)\otimes I_B + \Ib\otimes \Hb+ \Hsb
\\
H_S(t) &=& v_x(t)\sig_x+v_y(t)\sig_y
\\
\Hsb &=& \sig_z\otimes B
\eea
\eeq
The bath Hamiltonian $\Hb$ and the bath operator $B$ are constant but uncertain over any gate time $\tf$. 
From the definitions in \cref{eq:Gint} the corresponding 
interaction uncertainty Hamiltonian is,
\beq[eq:Hisb example]
\Gsb(t)=\widetilde S(t) \otimes \widetilde B(t)
\eeq
with interaction terms,
\beq[eq:Bi Si]
\widetilde{S}(t)=\Us(t)^\dag\sig_z\Us(t),\
\widetilde{B}(t)=\Ub(t)^\dag B\Ub(t)
\eeq
where $\Us(t)$ is the solution of \cref{eq:Usnom} and $\Ub(t)$ of \cref{eq:Ub}.

\subsection{Uncertainty bounds}
Assuming no knowledge of $\Hb$ and $B$, 
and only knowing from \cref{eq:ex const} that the bath couples to the system via $\sig_z$,
an obvious choice with this limited knowledge is to set the robustness measure 
for optimization to be as defined in \cref{eq:Jrbst set 1}, 
\beq[eq:Jrbst numex]
\Jrbst=\normsm{\avg{\Us^\dag\sig_z\Us}}
\eeq
Suppose it is possible to simultaneously make $\Jrbst=0$ and $\Fnom=1$. Then from Theorem~\ref{avgThm}
the bounding terms that make up $\tf\Omeff$ 
as defined in \cref{eq:Omunc def} become, 
\beq[eq:bnds Jzero]
\beasep{1.25}{rl}
\Omunc &= \norm{B}
\\
\Omavg &= \normsm{\avg{\widetilde S\otimes(\widetilde B-B)}}
\\
\dOmavg &= {\ds\max_t}\normsm{\widetilde S(t)\otimes\widetilde B(t)-\avg{\widetilde S\otimes(\widetilde B-B)}}
\eea
\eeq
If the bounding values \cref{eq:bnds Jzero} are known or learned, then the effective time-bandwidth product $\tf\Omeff$ 
and corresponding infidelity bound $1-\Flb$ can be calculated from Theorem~\ref{avgThm}. 

It is also worth noting that 
when $\Jrbst = 0$ (equivalently, $\avg{\Us^\dag\sig_z\Us} =0$),
it follows from Roth's lemma~\cite{Roth:1934aa},
$\overrightarrow{ A B C } = \left( C^{T} \otimes A \right) \vec{B}$,
that
the $4\times 1$ vector $\vec\sig_z$ must be in the nullspace of the
$4\times 4$ matrix,
\beq[eq:Acal]
\Acal = \frac{1}{\tf}\int_0^\tf \Us(t)^T \otimes \Us(t)^\dag dt
\eeq
This effect is verified in the numerical example to follow.

\subsection{Robust control optimization}
The control design goal is
to make the Hadamard gate: $\Ws=(\sig_x+\sig_z)/\sqrt{2}$.
With a magnitude constraint of $v_{\max}$  placed on the controls, a 
robust control candidate that makes both $1-\Fnom$ and $\Jrbst$ be $\approx 0$ 
is found by solving a single-stage optimization \cref{eq:onestage}
for controls $\{v_x(t),v_y(t),t\in[0,\tf]\}$ from,
\beq[eq:opt1qu const0]
\beasep{1.25}{ll}
\mbox{minimize}&
1-\Fnom + \lam \Jrbst
\\
\mbox{subject to}&
\Fnom =
|\trace(\Ws^\dag U_S(\tf))/2|^2
\\&
\Jrbst = \norm{\avg{\Us^\dag\sig_z\Us}}
\\&
|v_{x,y}(t)| \leq v_{\max}
\eea
\eeq
An interesting aspect of the optimization form is that except for the assumption that the bath is coupled via $\sig_z$,
\emph{no specific bath knowledge is required.} In addition, annihilating $\Jrbst$ would also reduce the impact of any constant coherent errors dependent on $\sig_z$.

The optimization is performed with the final time normalized to $\tf=1$, 
$N_{\rm pwc}=5$ piecewise constant (PWC) control pulses,
$v_{\max}=7.5$, and
$\lam=0.1$. The time-average of the interaction
Hamiltonian $\widetilde S(t)$ is approximated in discrete time by,
\beq[eq:avg numex]
\avg{\Us^\dag\sig_z\Us} 
\approx 
\frac{1}{\Navg}\sum_{k=1}^\Navg U_S(t_k)^\dag\sig_z U_S(t_k)
\eeq
Setting $\Navg=25$ results in $\Navg/\Npwc=5$ samples per pulse. This yields the $v_x$ and $v_y$ that define the robust control solution as two sequences of $5$ pulses. Theorem~\ref{avgThm} guarantees that the resulting infidelity will lie below the bound. 
The control pulses shown in \cref{had} 
achieve a nominal infidelity of $1-\Fnom = 1.81\times 10^{-7}$  
and a robustness measure $\Jrbst = 2.84\times 10^{-9}$. 
This low value of $\Jrbst$ indicates a very close proximity
to the nullspace defined by $\Acal$ from \cref{eq:Acal}.
The largest control magnitude is 5.70 in the $x$ channel, 
well within the constraint $v_{\max}=7.5$. For a gate time of $\tf=50$ nsec, 
the largest control magnitude would be 114 Mhz.

Although not shown, repeating the optimization from many random starts, all result in different pulse sequences with different performance. 
However, all return $\Fnom\approx 1$ and $\Jrbst\approx 0$. Additionally,
all provide similar performance in simulations when evaluated with the bath characteristics described below. However, we expect that incorporating additional information about the bath and allowing for different pulse shapes (e.g., Gaussian) has the potential to further reduce the infidelity and increase robustness.

\subsection{Performance evaluation}
To evaluate performance of a robust control from \cref{eq:opt1qu const0},
the worst-case lower bound $1-\Fwclow$ from \cref{eq:low} is computed with
the unknown bath uncertainties $(\Hb,B)$ modeled as combinations of qubits composed of Pauli matrices. 

Although many possible variations can be considered, \eg, bilinear coupling terms, spin baths as in quantum dots~\cite{sousa:115322}, \etc,
for illustrative purposes,
two instances are used to evaluate the robust control. The first is where $(\Hb,B)$ commute, and both 
are linear combinations of isolated $\sig_x$ terms. In the second $(\Hb,B)$ do not commute, with
$\Hb$ a linear combination of only $\sig_x$ terms and $B$ a linear combination of only $\sig_z$ terms.
For $\qb$ bath qubits (resulting in bath  
dimension $\nb=2^\qb$) the two cases are:
\beq[eq:HBB aligned]
\bea{ll}
\mbox{commuting}
&
\left\{
\beasep{1.25}{ccl}
H_{Bx} &=& \sum_{b=1}^\qb h_x^b \sig_x^b 
\\
B_x &=& \sum_{b=1}^\qb g_x^b \sig_x^b 
\eea
\right.
\eea
\eeq
and
\beq[eq:HBB unaligned]
\bea{ll}
\mbox{not commuting}
&
\left\{
\beasep{1.25}{ccl}
H_{Bx} &=& \sum_{b=1}^\qb
h_x^b \sig_x^b 
\\
B_z &=& \sum_{b=1}^\qb
g_z^b \sig_z^b 
\eea
\right.
\eea
\eeq
The $(h,g)$ coefficients are chosen 
randomly to restrict the range of $\norm{\tf B}$ and $\norm{\tf H_{B}}$
to the following sets of values in radians:
\beq[eq:eigs HBB]
\beasep{1.25}{ccc}
\norm{\tf B} &\in& \{0.15,0.3\}
\\
\norm{\tf H_{B}} &\in& [0.05,2]
\eea
\eeq
If the bath terms were actually commuting as indicated by \cref{eq:HBB aligned}, 
then $[H_B,B]=0$ so that the interaction-picture bath operator would be a constant, specifically,
\beq[eq:Bix align]
\tilde B_x(t) = U_{Bx}(t)^\dag B_x U_{Bx}(t)=B_x 
\eeq
This holds for isolated commuting terms in either $y$ or $z$ as well.
From the form of \cref{eq:Hisb example}, if the robust optimization results in
$\Fnom=1$ and $\avg{\Us^\dag\sig_z\Us}=0$, it then follows from \cref{eq:bnds Jzero} that 
the effective uncertainty is equal to the intrinsic uncertainty,
\beq[eq:intrinsic unc]
\tf\Omeff=\norm{\tf B_x}
\eeq
The blue and red plots show the worst-case lower bound $1-\Fwclow$ from \cref{eq:Fnom1}. 
The blue circles in \cref{had} (left and middle) correspond to when the bath Hamiltonians are commuting
as in \cref{eq:HBB aligned}, each being composed, respectively, of $\qb=2\ (\nb=4)$ 
and $\qb=6\ (\nb=64)$ uncertain linear combinations of $\sig_x$.
The resulting uncertainty error is at 0.15 and 0.3 radians reflecting exactly the two 
values in \cref{eq:eigs HBB}: 
$\norm{\tf\Omeff}=\norm{\tf B_x}\in\{0.15,0.3\}$. 
When commuting, the infidelity is unaffected by 
the range of $\norm{\tf\Hb}$ \cref{eq:eigs HBB}.  

When $(\Hb,B)$ do not commute, as in \cref{eq:HBB unaligned}, there is a clear dependence on the range of 
$\norm{\tf\Hb}$ as well as
a noted increased robustness with higher bath dimension. 
The light red and light blue lines in \cref{had} (left and middle) result from $N_{\rm mc}=100$ random $(h,g)$ 
coefficient samples. The corresponding effective uncertainty 
measure increases along with an increase in infidelity.
However, significantly more robustness is retained for the larger bath dimension
despite $\norm{\tf\Hb}$ varying over the same range \cref{eq:eigs HBB}. 

The example reveals the interesting phenomenon that a larger bath dimension yields more robustness. 
One possible explanation is that it takes longer for any
cumulative effects to return to the 
system, thus a slow recovery time 
with respect to the gate time.
Conversely, a small bath dimension can have a relatively fast recovery time and thus cause more disruption.

\section{Concluding Remarks and Outlook}
\label{sec:conclude}

Theorem~\ref{avgThm} settles a
long-standing question in robust quantum control: \textit{How good can
a quantum gate be if every error is ``either known or unknown but bounded''?}
By expressing the worst-case infidelity solely as a function of the
dimensionless time-bandwidth product $\tf\Omeff$, Theorem~\ref{avgThm} exposes a
fundamental property of quantum dynamics: an \textit{intrinsic}
robustness that cannot be outperformed but can, in favorable cases, be
\textit{attained}.  \Cref{lim} shows this performance limit
for a single gate; \cref{lim1} generalizes the picture across
gate durations relevant to near-term hardware.

The infidelity upper bound $1-\Flb$ is deliberately agnostic to the specific route
taken to reduce $\tf\Omeff$.  It applies whether uncertainty is
suppressed through better materials, refined fabrication, dynamical
decoupling, pulse-shaping, \etc, or any combination
thereof.  Theorem~\ref{avgThm} asks only that the \textit{nominal}
(uncertainty-free) model realizes the target unitary.  If the available
design and control degrees of freedom can eliminate (or greatly diminish)
every time-averaged interaction term that couples to
uncertainty, the gate automatically achieves the minimum infidelity
permitted by quantum mechanics under the stated assumptions.

A direct corollary of the theory is a two-objective optimization
scheme: minimize (i) the nominal infidelity and (ii) the
time-averaged error generators derived from the uncertainty model.  A
practical advantage is that \textit{both} objectives depend only on the
uncertainty-free evolution; no Monte Carlo sampling over
high-dimensional bath realizations is required during pulse search.
\Cref{sec:numex} demonstrated that this strategy can efficiently locate
controls whose observed errors are close to the theoretical bound across
numerous randomly sampled bath parameters.

The entire framework inherits a key principle from classical robust
control~\cite{Zames:1966,DesVid:1975,LMI94,ZhouDG:96}: explicit
set-based uncertainty modeling.  Once a model set is posited, the bound
in Theorem~\ref{avgThm} becomes \textit{experimentally falsifiable} in
Popper's sense~\cite{popper2002logic}.  Any measured fidelity that falls
\textit{below} $\Flb$ signals that the true system lies
outside the assumed set, thereby falsifying the model and prompting a
re-examination of device physics, control assumptions, or the bath model~\cite{SafonovT:97,KosutBDA:97}.  Conversely, repeated agreement between
experiment and bound certifies consistency with the model.

Determining a credible $\tf\Omeff$ for state-of-the-art devices will
likely require specialized identification and validation protocols
\cite{id4c:92,KosutLB:92,modval:94}.  Data-driven estimation of
uncertainty magnitudes, e.g., via randomized benchmarking,
noise spectroscopy, or recently deterministic benchmarking
experiments~\cite{tripathi2024-DB}, can feed directly into the
time-bandwidth metric without demanding a full microscopic model.

As remarked in the introduction, two paths may be taken for experimental gate performance tests. In either case, the bounding curve is used to predict expected performance, and the range of the effective uncertainty level $\tf\Omeff$ and actual infidelity $1-F$ needs to be determined from the laboratory data. At that point, the theoretical infidelity bound $1-\Flb$ can be used as a comparative performance metric.

Not every imperfection is addressable within the present bound.
Catastrophic errors, such as qubit loss or full state erasure, lie
outside the ``unknown but bounded'' paradigm and must still be handled by
quantum error correction~\cite{Wu:2022aa}.

Our results have implications for fault-tolerance thresholds. Namely, 
if experiments can quantify $\tf\Omeff$ for a particular processor, the
curve in \cref{lim} immediately reveals whether that platform's
intrinsic error floor is below the threshold demanded by fault-tolerance estimates.
Hence the bound serves as both a design target and a
benchmark: it can inform hardware engineers of the uncertainty
reduction required for large-scale quantum processors, and it provides
control theorists with an objective function whose minimization
guarantees performance improvements.

We envision several future directions.
\begin{itemize}
\item Tightness analysis: While the bound is already within
one order of magnitude of numerically observed errors, a systematic
study of tightness across higher-dimensional gates and strongly
non-Markovian environments would clarify the gap between worst-case
theory and typical performance.
\item Extension to multi-gate sequences: Folding the bound
into a whole-circuit analysis could connect single-gate robustness
directly with logical error rates, thus complementing fault-tolerance
simulations.
\item Integration with adaptive and measurement-based
control: The present work excludes measurement-based feedback; combining the
time-bandwidth metric with the latter could
further improve performance.
\item Alternative uncertainty measures: Recasting the theory
in terms of correlation functions will cover baths that violate finite-norm
assumptions.

\item Uncertainty estimation: The bound offers the framework and motivation to develop experimentally quantifiable hardware metrics that provide falsifiable benchmarks. 
\end{itemize}

In summary, Theorem~\ref{avgThm} provides a unifying perspective through which
to view robustness, control design, and experimental validation.
Combined with ongoing advances in device fabrication and pulse
optimization, it lays a quantitative foundation for closing the gap
between current noisy processors and future fault-tolerant quantum
computers.

\section*{Acknowledgments} 

For RLK and DAL, this material is based upon work supported by, or in part by, the
U. S. Army Research Laboratory and the U. S. Army Research Office
under contract/grant numbers W911NF2310307 and W911NF2310255. For RLK and
HR, support from the U.S. Department of Energy (DOE) under STTR Contract
DE-SC0020618, and HR under DOE grant DE-FG02-02ER15344 for control landscape concepts.

\appendix

\section{Fidelity and Distance}
\label{sec:Fpsimin}

\subsection{Some basic inequalities} 
Let $U$ be an arbitrary $\n\times \n$ unitary and $\ket{\psi}$ an arbitrary, normalized pure state. Set 
\beq[eq:Delta]
E\equiv U-I  \ , \quad z(\psi)\equiv\bra{\psi}U\ket{\psi}
\eeq

\begin{lem}[Worst-case fidelity lower bound]
\label{lem:1}
\beq[eq:ineq]
F(\psi)\equiv |\bra{\psi}U\ket{\psi}|
\geq
\max\bigl(1-\frac{1}{2}\norm{E}^2,0\bigr)
\eeq
where $\normsm{E}\equiv\sup_{\norm{x}_2=1} \norm{E x}_2$ denotes the induced $2$-norm.
\end{lem}
\begin{proof}
We have $|z|\ge \mathrm{Re}(z)$. But $\norm{E\ket{\psi}}^2 = \bra{\psi}E^\dagger E\ket{\psi} = 2-2\mathrm{Re}(z)$, i.e., $\mathrm{Re}(z) = 1-\frac{1}{2}\norm{E\ket{\psi}}^2$. Using $\norm{E\ket{\psi}}^2 \leq \norm{E}^2$, 
\cref{eq:ineq} 
follows.
\end{proof}

\begin{lem}[Average fidelity lower bound]
\label{lem:frobenius}
Define the state averaged gate fidelity
\beq[eq:F-ave-def]
\Fave\equiv \int \bigl|\bk{\psi}{U}\bigr| d\psi,
\eeq
where the integral is over the Haar measure on pure states. Then
\beq[eq:favg_bound]
\Fave \ge \max\bigl(1 - \frac{1}{2\n}\|E\|_{F}^{2},0\bigr)
\eeq
where $\|E\|_{\mathrm{F}} \equiv \sqrt{\Tr(E^\dag E)}$ denotes the Frobenius norm.
\end{lem}

\begin{proof}
Recall Jensen's inequality: for any convex function $f$ 
\beq
\mathbb{E}[f(X)]\ge f(\mathbb{E}[X])
\eeq
where $\mathbb{E}$ denotes the expectation value of the random variable $X$. In the case of interest to us, $f$ is the modulus (a convex function on $\mathbb{C}$), $\mathbb{E}$ is the Haar average, and the random variable is $\bk{\psi}{U}$. Thus, we have
\beq[]
\Fave = \int\abs{\bk{\psi}{U}}\,d\psi
   \ge
   \abs{\int \bk{\psi}{U} d\psi} 
\eeq
Now observe that 
\beq
\begin{aligned}
\Tr\left[M\int\ketbra{\psi}d\psi\right] &= \int\Tr[M\ketbra{\psi}]d\psi = \int\bk{\psi}{M} d\psi \\
&= \frac{1}{\n}\Tr(M)
\end{aligned}
\eeq
for any fixed operator $M$, where the last equality follows since the map
$M\mapsto \int\bk{\psi}{M} d\psi$ is linear and unitarily invariant.
Invariance forces the functional to be a scalar multiple of $\Tr(M)$; evaluating at $M=I$ fixes the multiple to $1/\n$. Thus, replacing $M$ with $U$, we obtain
\beq[]
\Fave \ge \frac{1}{\n}\abs{\Tr(U)} \ge \frac{1}{\n}\mathrm{Re}\Tr(U)
\eeq
On the other hand, $U$ being unitary gives $\|E\|_{F}^{2}=2\n-2\mathrm{Re}\Tr(U)$, so that, finally,
\beq[]
\Fave \ge \frac{1}{\n}(\n-\frac12\|E\|_{F}^{2}) = 1 - \frac{1}{2\n}\|E\|_{F}^{2} \equiv \Favelow
\eeq
\end{proof}

\subsection{Lower bound}

If the target unitary $\Ws$ is achieved by the nominal
(uncertainty-free) system, then from \cref{eq:FstateR} at the
final-time, fidelity only depends on the interaction-picture unitary $\Ui(\tf)$.
Stated formally as,
\beq[eq:Fid nomactR]
\begin{aligned}
&\left\{
  \beasep{1.25}{c}
  \Fnom = 1
  \\
  \mbox{equivalently}
  \\
  \Us(\tf)= \phi \Ws,\ |\phi|=1 
  \eea
\right\}
\quad\Rightarrow\\
&\left\{
\beasep{1.5}{ll}
&
F(\psin) =|\psiinc \Ui(\tf) \psiin|
\\
&
\Ui(\tf) = (\Ws\otimes \Ub(\tf))^\dag U(\tf)
\eea
\right.
\end{aligned}
\eeq
%

The basic inequalities \cref{eq:ineq} and \cref{eq:favg_bound} immediately establish the
lower-bound in \cref{eq:Fnom1}, \ie,
\beq[eq:UJfid lb]
\beasep{2.5}{ccc}
\Fwc \geq \Fwclow \equiv \max\bigl(1-\frac{1}{2}\|\Ui(\tf)-I\|^2,0\bigr)
\in[0,1]\\
\Fave \geq \Favelow \equiv \max\bigl(1-\frac{1}{2{\n}}\|\Ui(\tf)-I\|_{\mathrm{F}}^2,0\bigr)
\in[0,1]
  \eea
\eeq

A standard norm inequality between the Frobenius and induced $2$-norm is
$\|A\|_{\mathrm{F}}\leq\sqrt{\n}\|A\|$ for any operator $A$. As a result, $\Favelow \geq \Fwclow$,
as expected.
Using the eigenvalue decomposition,
\beq[eq:eigR]
\Ui(\tf) = V e^{i\tf\Om}V^\dag,\quad\Om =\diag(\om),\ \om\in\Rbf^{\n}
\eeq
and substituting into the lower bound functions in \cref{eq:UJfid lb}
gives,
\beq[eq:omFlb]
\begin{aligned}
\Fwc &\geq \Fwclow
=\max\bigl(1-\frac{1}{2}\norm{e^{i\tf\Om}-I}^2,0\bigr)\\
&
=\max\bigl(\min_{k\in\{1,\n\}}\cos(\tf\om_k),0\bigr)\\
\Fave &\geq \Favelow
=\max\bigl(1-\frac{1}{2\n}\norm{e^{i\tf\Om}-I}_{\mathrm{F}}^2,0\bigr)\\
&=\max\bigl(\frac{1}{\n}\sum_{k=1}^{\n}\cos(\tf\omega_{k}),0\bigr)\\
\end{aligned}
\eeq
where we used $|e^{ix}-1|^2 = 2(1-\cos x)$.
Comparing $\Fwc$ with the lower bound function $\Fwclow$ for fidelity
errors $1-\Fwc\in[10^{-6},10^{-2}]$ results in small relative errors
$\Fwc/\Fwclow-1\leq0.001$. This small error holds over a range of
dimensions $\n$ and various eigenvalue distributions $\om\in\Rbf^{\n}$
satisfying $\om_k\tf\leq \cos^{-1}\Fwclow$.
As shown in \cref{sec:rbst opt}, calculating $\Fwc$ or $\Fwclow$ is
needed only for evaluation, not for optimization.  Clearly $\Fwclow$
is a good approximation for $\Fwc$ in the fidelity range of interest.
To make full use of \cref{eq:UJfid lb} it remains to bound
$\|\Ui(\tf)-I\|$, the deviation from identity of the final-time
interaction-picture unitary, equivalently, the deviation of the system unitary
from the uncertainty-free ideal target. In the next section we show
how to use knowledge about the uncertainty Hamiltonian
$\Hunc(t)\in\Hcalunc$ to bound robust performance.

\subsection{Calculating worst-case fidelity}

Following \cref{eq:FstateR}, the worst-case fidelity
$\Fwc=\min_{\psin}|\psiinc A \psiin|$ with the $\n\times \n$ unitary
$A=\Ws^\dag \Us(\tf)\otimes I_B)\Ui(\tf)$, can be found from the
equivalent convex optimization,
\beq[eq:cvx rho]
\beasep{1.25}{ll}
\mbox{minimize} & |\trace(A\rho)|
\\
\mbox{subject to}& \rho \geq 0,\ \trace \rho=1
\eea
\eeq
where $\rho$ can be an arbitrary mixed state. The resulting optimal density matrix
$\rho_{\rm opt}$ determines the minimum (worst-case) fidelity as,
$\Fwc = |\trace(A\rho_{\rm opt})|$.


\section{Proof of Robust Performance Limit} 
\label{sec:avgThm}

Under the same conditions for which \cref{eq:Fnom1} and \cref{eq:UJfid lb} hold, 
the fidelity is bounded below by,
\beq
\begin{aligned}
F(\psin) &\geq \Fwc \ge \Fwclow = \max\bigl(1-\frac{1}{2}\|\Ui(\tf)-I\|^2,0\bigl)\\
&\geq \Flbw \ge 0
\end{aligned}
\eeq
provided that,
\beq[eq:Ebndtwo]
\|\Ui(\tf)-I\| \leq \sqrt{2\left(1-\Flbw\right)}
\in [0,\sqrt{2}]\ , \quad \Flbw\in [0,1]
\eeq
Similarly,
\beq[eq:Fbndave]
\Fave \ge \Favelow = \max\bigl(1-\frac{1}{2{\n}}\|\Ui(\tf)-I\|_{\mathrm{F}}^2,0\bigl)
\geq \Flba \ge 0
\eeq
provided that,
\beq[eq:Ebndtwoave]
\begin{aligned}
&\|\Ui(\tf)-I\|_{\mathrm{F}} \leq \sqrt{2\n\left(1-\Flba\right)}
\in [0,\sqrt{2\n}]\ , \\
&\qquad \Flba\in [0,1]
\end{aligned}
\eeq
$\Flbw$ and $\Flba$ are defined below, in \cref{eq:flb bnd}.

To bound the left-hand side of \cref{eq:Ebndtwo,eq:Ebndtwoave} we first apply the form of the
standard state transformation for averaging analysis described in
\cite[\S V.3]{Hale:80} and \cite{KosutBR:2022} (periodicity, usually
assumed, is not needed here).  Set,
\beq[eq:r2v]
\beasep{1.25}{rcl}
\Ui(t) &=& (I+K(t))V(t)
\\
K(t) &=& -i\int_0^t\left(\Hi(\tau)-\avg{\Hi}\right)d\tau
\eea
\eeq
with $\Hi(t)$ from \cref{eq:Gint}.
For $t\in(0,T)$, $V(t)$ is the solution of,
\beq[eq:vode]
\bea{rcl}
\dot V(t) &=& -i\Del(t)V(t),\
V(0)= I
\\
\Del(t) &=& (I+K(t))^{-1}(\Hi(t)K(t) + \avg{\Hi})
\eea
\eeq
Observe that
$K(0)=K(\tf)=0$ which implies that $V(0)=\Ui(0)= I$ and
$V(\tf)=\Ui(\tf)$. Since $V(0)= I$, deviations of $V(\tf)$ from identity determine the
limit (via the method of averaging) of robust performance. Integrating
\cref{eq:vode} gives the error for any $t\in[0,\tf]$ as,
\beq
E(t) = V(t)- I = -i\int_0^t\Del(s)ds-i\int_0^t\Del(s)E(s)ds
\eeq
Bounding the error in any fixed unitarily invariant (hence sub-multiplicative) norm $\|\cdot\|_{\rm{ui}}$ yields,
\beq[eq:Ebnd]
\norm{E(t)}_{\rm{ui}} \leq \int_0^t\norm{\Del(s)}_{\rm{ui}}ds
+ \int_0^t\norm{\Del(s)}_{\rm{ui}} \norm{E(s)}_{\rm{ui}}ds
\eeq

\begin{lem}
\label{lem:K-inv-bound}
Let $K$ be anti-Hermitian on a $\n$-dimensional Hilbert space. Then for the induced $2$-norm (maximum singular value)
\beq[eq:invopnorm]
\normsm{(I+K)^{-1}}\leq 1
\eeq
whereas for the Frobenius norm
\beq[eq:invFrob]
  \lVert (I+K)^{-1}\rVert_{F} \;\le\; \sqrt{\n} 
\eeq
Both bounds are optimal in the sense that no smaller universal upper bound holds for all anti-Hermitian $K$.
\end{lem}

\begin{proof}
Suppose $K$ is anti-Hermitian, i.e., $K^\dagger = -K$. 
Since $K$ is anti-Hermitian, we may write $K = i H$ where $H= - i K$ is Hermitian. 

\paragraph*{Induced $2$-norm:} We have $K^2 = - H^2$, hence $I - K^2 = I + H^2$,
where $H^2$ is positive-semidefinite. Therefore $I + H^2$ is also positive-semidefinite and satisfies
$\bra{x} \bigl(I + H^2\bigr) \ket{x} \ge \|\ket{x}\|^2$.
Note that for any vector $\ket{x}$ (not necessarily normalized),
\beq
\begin{aligned}
\|(I+K)\ket{x}\|^2 &=
\bra{x}(I+K)^\dagger (I+K)\ket{x}\\
&
=
\bra{x}(I-K)(I+K)\ket{x}\\
&=
\bra{x}\bigl(I - K^2\bigr)\ket{x}.
\end{aligned}
\eeq
Thus, $\|(I+K)\ket{x}\| \ge \|\ket{x}\|$ for every $\ket{x}$. This is equivalent to $\sigma_{\min}(I+K) = \inf_{|x\rangle\neq 0}\frac{\|(I+K)|x\rangle\|}{\||x\rangle\|}\ge 1$, i.e., the smallest singular value of $(I+K)$ is at least $1$. This implies that $(I+K)$ is invertible, and
\beq
\bigl\|(I+K)^{-1}\bigr\|
=
\frac{1}{\sigma_{\min}(I+K)}
 \le 1 .
\eeq

Optimality: consider $K = i \alpha I$ with real $\alpha$. Then $K$ is clearly anti-Hermitian, and $I + K
=(1 + i \alpha) I$, whose inverse is $
(I + K)^{-1}
=
\frac{1}{1+i \alpha} I$,
and
\beq
\|(I + K)^{-1}\|
=
\frac{1}{|1+i \alpha|}
=
\frac{1}{\sqrt{1+\alpha^2}}
 \le 1.
\eeq
As $\alpha \to 0$, the quantity $\|(I+K)^{-1}\|$ approaches $1$. Since the bound must hold for any $K$, we conclude that $\|(I+K)^{-1}\|\le 1$ is sharp.

\paragraph*{Frobenius norm:} Diagonalize 
\beq
H = V \operatorname{diag}(h_{1},\dots,h_{\n}) V^{\dagger},\quad h_{j}\in\mathbb R
\eeq    
Unitary invariance of the Frobenius norm gives
\beq[]
  \norm{(I+K)^{-1}}_{F}^{2}
    = \sum_{j=1}^{\n} \frac{1}{\abs{1 + i h_{j}}^{2}}
    = \sum_{j=1}^{\n} \frac{1}{1 + h_{j}^{2}}
    \le \sum_{j=1}^{\n} 1
    = d 
\eeq
Taking square roots yields \cref{eq:invFrob}.

Optimality: choose $K = 0$.  Then $(I+K)^{-1}=I$ and 
$\lVert I\rVert_{F} = \sqrt{\Tr(I)} = \sqrt{\n}$, saturating
the bound.  Hence $\sqrt{\n}$ is the smallest constant valid for every
anti-Hermitian $K$.

\end{proof}

Using \cref{eq:vode,eq:invopnorm} we now have, in the worst-case setting,
\beq[eq:Dbnd1]
\normsm{\Del(t)} \leq \|\Hi(t)\|\normsm{K(t)}+\|\avg{\Hi}\|
\eeq

For the average-case setting, we note that since the Frobenius norm is the $\ell_2$-norm of the singular values, we have $\norm{A}_{\mathrm{F}} \le \sqrt{r} \norm{A}$ where $r=\rank(A)$. Therefore, $\norm{AB}_{\mathrm{F}} \le \min\bigl(\sqrt{\rank(A)},\sqrt{\rank(B)}\bigr) \norm{A}\norm{B}$, and
\beq[eq:Dbnd1-F]
\normsm{\Del(t)}_{\mathrm{F}} \leq \kappa \bigl(\|\Hi(t)\|\normsm{K(t)}+\|\avg{\Hi}\|\bigr)
\eeq
where
\beq
\kappa^2 = \min\bigl(\rank[(I+K(t))^{-1}],\rank[\Hi(t)K(t) + \avg{\Hi}]\bigr)
\eeq
It may be difficult to estimate $\kappa$ in practice. However, we can always use the looser bound $\kappa\le \n$, which is also what we obtain from \cref{eq:invFrob}. In this case, we have
\beq[eq:Dbnd1-Fd]
\normsm{\Del(t)}_{\mathrm{F}} \leq \sqrt{\n} \bigl(\|\Hi(t)\|\normsm{K(t)}+\|\avg{\Hi}\|\bigr)
\eeq

Using the bounds defined in \cref{eq:Omunc def},
\beq[eq:Om bnd]
\|\Hi(t)\| \leq \Ommax,
\quad
\|\avg{\Hi}\| \leq \Omavg,
\quad
\|\Hi(t)-\avg{\Hi}\| \leq \dOmavg
\eeq
\cref{eq:Dbnd1,eq:Dbnd1-F} can be written as,
\beq[eq:Dbnd2]
\begin{aligned}
\normsm{\Del(t)} &\leq \bigl(
\Ommax~\normsm{K(t)}+\Omavg \bigr)\\ 
\normsm{\Del(t)}_{\mathrm{F}} &\leq \kappa\bigl(
\Ommax~\normsm{K(t)}+\Omavg \bigr)
\end{aligned}
\eeq
A bound on $K(t)$ can be found in two ways. First,
\beq[eq:Kbnd1]
\norm{K(t)} = \norm{\int_0^t (\Hi(s)-\avg{\Hi})ds}
\leq \dOmavg~t
\eeq
Second, replace $\int_0^t(\Hi(s)-\avg{\Hi})ds$ with
$\int_0^T(\Hi(s)-\avg{\Hi})ds-\int_t^T(\Hi(s)-\avg{\Hi})ds$. Since the
first of these terms is zero, the bound is then,
\beq[eq:Kbnd2]
\beasep{1.25}{rcl}
\normsm{K(t)}
&\leq& \dOmavg(\tf-t)
\eea
\eeq
Altogether, using the minimum bound on $\norm{K(t)}$ for
$t\in[0,\tf]$,
\beq[eq:Kbnd]
\norm{K(t)} \leq \dOmavg~\bet(t),
\quad\bet(t)=
\left\{
\bea{ll}
t & t < \tf/2
\\
\tf-t & t > \tf/2
\eea
\right.
\eeq
Combining with \cref{eq:Dbnd2},
\beq[eq:Dbnd]
\normsm{\Del(t)} \leq \del(t) \equiv \Ommax~\dOmavg~\bet(t)+\Omavg\ , \quad 
\normsm{\Del(t)}_{\mathrm{F}} \leq \kappa \del(t) 
\eeq
Then \cref{eq:Ebnd} becomes,
\beq[eq:Ebnd1]
\norm{E(t)} \leq c(t) + \int_0^t \dot c(s)\norm{E(s)}ds
\quad
\left\{
\beasep{1.25}{l}
c(t) = \int_0^t \del(s)ds
\\
\dot c(t) = \del(t)
\eea
\right.
\eeq
and
\beq[eq:Ebnd1-F]
\norm{E(t)}_{\mathrm{F}} \leq \kappa\Big(c(t) + \int_0^t \dot c(s)\norm{E(s)}_{\mathrm{F}} ds\Big)
\eeq
Since $c(0)=0$, we can
use the version of the Bellman-Gronwall Lemma in \cite{CoddLev:55}
which gives the bound,
\beq[eq:Ebnd2]
\begin{aligned}
\norm{E(t)} &\leq \int_0^t\dot c(s)
~\exp{\int_s^t \dot c(\tau)d\tau}\ ds\\
\norm{E(t)}_{\mathrm{F}} &\leq \kappa \int_0^t\dot c(s)
~\exp{\kappa \int_s^t \dot c(\tau)d\tau}\ ds
\end{aligned}
\eeq
Performing the indicated integrations evaluated at $t=\tf$ and using
$V(\tf)=\Ui(\tf)$,
\beq[eq:Ebnd3]
\beasep{1.5}{l}
\norm{E(\tf)} = \norm{V(\tf)- I}=\|\Ui(\tf)- I\|
\leq e^{c(\tf)}-1
\\
\norm{E(\tf)}_{\mathrm{F}} = \norm{V(\tf)- I}_{\mathrm{F}}=\|\Ui(\tf)- I\|_{\mathrm{F}}
\leq e^{\kappa c(\tf)}-1
\\
c(\tf) = \tf\Omavg + (\tf\Ommax)(\tf\dOmavg)/4
\eea
\eeq
To ensure \cref{eq:Ebndtwo,eq:Ebndtwoave} hold requires that,
\beq[eq:expc bnd] 
\begin{aligned}
e^{c(\tf)}-1 &=
\sqrt{2\left(1-\Flbw\right)} \\
e^{\kappa c(\tf)}-1 &=
\sqrt{2\n\left(1-\Flba\right)}
\end{aligned}
\eeq
or equivalently,
\beq[eq:flb bnd]
\begin{aligned}
\Flbw &= \max\left(1-\frac{1}{2}\bigl(e^{c(\tf)}-1\bigr)^2,0\right)\\
\Flba &= \max\left(1-\frac{1}{2{\n}}\bigl(e^{\kappa c(\tf)}-1\bigr)^2,0\right)
\end{aligned}
\eeq
Rearranging terms gives, for the worst-case 
\beq[eq:Tdelbnd]
c(\tf) 
=
\ln\left(
1+\sqrt{
2\left(
1-\Flbw
\right)
}
\right)
\eeq
and for the average case
\beq[eq:Tdelbnd-F]
c(\tf)
=
\frac{1}{\kappa}\ln\left(1+\sqrt{2\n\left(1-\Flba\right)}\right)
\eeq
When the interaction-picture Hamiltonian time-average $\Omavg=\avg{\Hi} \neq 0$,
then the limit bound can be expressed in a variety of ways, for example,
as in Theorem~\ref{avgThm}, 
\beq[eq:Tdelbnd_1]
\tf\Omeff \equiv \sqrt{(\tf\Ommax)(\tf\dOmavg) + 4\tf~\Omavg} 
= 2\sqrt{c(T)}
\eeq
Since 
$\Flbw\in[0,1]$,
$\tf\Omeff$ is maximized when $\Flbw=0$.
Thus, in the worst-case setting
\beq[eq:TOmeff uppbnd]
0 \leq \tf\Omeff 
\leq
2\sqrt{
\ln
\left(
1+\sqrt{2}
\right)
}
=1.8776 \mbox{ radians}
\eeq
For example, with a gate time of $\tf=50$ nsec, $\Omeff \leq 37.55$ Mhz.
When the interaction-picture Hamiltonian time-average $\Omavg=\avg{\Hi}=0$,
then $\dOmavg=\Ommax$ and the limit bound becomes
$\tf\Omeff=\tf\Ommax$. 

In the average-case setting, on the other hand,
\beq[eq:TOmeff uppbnd-F]
0 \leq \tf\Omeff 
\leq
2\sqrt{\frac{1}{\kappa}\ln\left(1+\sqrt{2\n}\right)}
\approx 2\sqrt{\frac{n+1}{2\kappa}}
\eeq
for $\n=2^n$ in the case of a system of $n$ qubits. Note that the RHS approaches zero if $\kappa$ scales faster than $O(n)$, which is expected for most Hamiltonians. This points to a problem with the Frobenius norm bound. Evidently, explicitly bounding $\|\Delta\|_{\mathrm{F}}$ using the inequality in \cref{eq:Dbnd1-Fd}, which relates $\|\Delta\|_{\mathrm{F}}$ to the $2$-norm of the various Hamiltonians, makes the Frobenius norm bound too loose. We leave it as an open problem to tighten the Frobenius norm lower bound. 

Another way to state the problem, which is clear by comparing $\Flbw$ and $\Flba$ in \cref{eq:flb bnd}, is that $\Flbw\ge\Flba$ except for $\kappa=1$, the opposite of the expected ordering.
This means that our lower bound on the average fidelity is far from tight.
However, since by definition $\Fave \ge \Fwc$, and $\Fwc \ge \Flbw$, we can simply replace $\Flba$ by $\Flbw$, which is what we did in the statement of Theorem~\ref{avgThm}, while renaming $\Flbw$ as $\Flb$.

\section{Extensions to Uncertainty Model Framework}
\label{sec:extend}

\subsection{Summary}
A few extensions are briefly discussed which fit the uncertainty model
framework where each has a similar structure and resulting robust performance limit bounds: 
(i) Lindblad, (ii) ancilla,
(iii) multilevel systems, and (iv) crosstalk.
With some modifications,
the theoretical framework and performance bound Theorem~\ref{avgThm} can be
extended unchanged except for computing the time-bandwidth uncertainty
bound \cref{eq:Omunc def}.
In general for these latter three, the total dimension $\n$ 
defined in \cref{eq:SBsys} depends on what is
labeled there as ``system'' and ``bath.'' 
For the basic bipartite system $\n=\ns\nb$. 
Ancilla states of dimension $\na$ are typically added in a product state 
with both the system and bath, hence, $\n=\ns\na\nb$.  
For multilevel systems
with $\ne$ extra levels, $\n=(\ns+\ne)\nb$. 
In many implementations, other qubits, supposedly idle, in fact
cause unwanted interactions just by their proximity to the ``active''
qubits performing the required sequential logical operation. 
Referred to as ``crosstalk,'' 
the total dimension should include 
a sufficient number of the neighboring
controlled quantum states running in parallel during the operation
time. Thus the ``B''-system dimension \cref{eq:SBsys} is not just the
bath, but also the interference induced by these $\nq$ neighboring states,
resulting in $\n=\ns\nq\nb$. 

\subsection{Lindblad master equation}

As previously noted, the induced norm of bosonic bath Hamiltonians
diverges with bath dimension, \eg, for $\Balf(t)$ from \cref{eq:Hsbalf}, 
$\norm{\Balf(t)}\to\infty$ as $\nb\to\infty$. 
As argued, \eg, in \cite{DDPRA:2011}, this requires a different measure of uncertainty, 
\eg, based on input-state-dependent correlation functions. 
The Lindblad master equation, under suitable conditions, very well describes 
open system non-unitary evolution
in terms of rates computable using correlation functions~\cite{Breuer:book,rivas_open_2012,Majenz:2013qw}. Its range of validity is nevertheless restricted by the convergence of time-dependent perturbation theory, which is usually prescribed in terms of diverging quantities such as $\norm{\Balf(t)}$~\cite{Mozgunov:2019aa}. Therefore, the extension we present in this should \emph{not} be perceived as a complete solution to the problem of diverging operator norms.

Starting from \cref{eq:SBsys} and tracing out the bath, the $\ns\times\ns$ system density
matrix is,
\beq[eq:rhos]
\rho_S(t)=\Tr_B(\ketb{\psi(\tf)}{\psi(\tf)}),\ t\in[0,\tf]
\eeq
Under the assumption that the initial state is decoupled from the
bath, \ie, $\ket{\psi(0)}=\ket{\psi_S(0)}\otimes\ket{\psi_B(0)}$, the general
differential Lindblad form is,
\beq[eq:lind]
\beasep{1.5}{rl}
\dot\rho_S(t) &= -i[\bar{H}_S(t),\rho_S(t)]
+ \Lcal(\rho_S(t))
\\
\Lcal(\rho_S) &= \sum_{\ell=1}^{m} \gam_\ell \Lcal_{\ell}(\rho_S)
\\
\Lcal_\ell(\rho_S) &=
L_\ell\rho_S L_\ell^\dag-
\frac{1}{2}\left\{L_\ell^\dag L_\ell,\rho_S\right\}
\eea
\eeq
with $\bar{H}_S(t)=\Hs(t)+\Hscoh(t)$ as defined in \cref{eq:Ham}, but with $\Hscoh(t)$ including a component induced by the system-bath coupling known as the Lamb shift~\cite{Lidar200135,Majenz:2013qw}.
Here we have assumed that the Lindblad operators $L_\ell$ are constant; 
in general, they could be time-varying. When the rates $\gam_\ell$ are all nonnegative, \cref{eq:lind} is known as the Lindblad equation, and it describes Markovian dynamics. Otherwise, \cref{eq:lind} is a general quantum master equation that can describe non-Markovian dynamics~\cite{Rivas:2014aa}.
The limit bound Theorem~\ref{avgThm} encompasses the Lindblad form
by lifting the density matrix to the $d_S^2$-dimensional vector
$\vec\rho_S(t)$ .
The lifted (or ``vectorized'') state evolution version of
\cref{eq:lind} is governed by,
\beq[eq:vec rho evo]
\dot{\vec{\rho}}_S(t) = (-iA(t)+D)\vec{\rho}_S(t)
\eeq
with $d_S^2\times d_S^2$-dimensional matrices $A(t)$ and $D$ given by,
\beq[eq:mat rho evo]
\beasep{1.25}{ccl}
A(t) &=& I_S\otimes \bar{H}_S(t)-\bar{H}_S(t)^T\otimes I_S
\\
D &=& \sum_{\ell=1}^{m} \gam_{\ell}D_\ell
\\
D_\ell &=& L_{\ell}^*\otimes L_{\ell}
-\frac{1}{2}\left(
I_S\otimes L_{\ell}^\dag L_{\ell}
+(L_{\ell}^\dag L_{\ell})^T\otimes I_S
\right)
\eea
\eeq
Define the $d_S^2\times d_S^2$-dimensional interaction matrix $V(t)$ via the
lifted state $\vec\rho(t)$ as,
\beq[eq:R Phi]
\beasep{1.25}{rcl}
\vec\rho(t) &=& \Phi_S(t)V(t)\vec\rho_0
\\
\Phi_S(t) &=& \Us(t)^*\otimes \Us(t)
\eea
\eeq
with uncertainty-free unitary $\Us(t)$ from \cref{eq:Usnom} and $V(t)$ from,
\beq[eq:R evo]
\beasep{1.25}{rcl}
\dot V(t) &=&
\Big(\sum_{\ell=1}^m\gam_\ell G_\ell(t)\Big)V(t),\ V(0) = I_{d_S^2}
\\
G_\ell(t) &=& \Phi_S(t)^\dag D_\ell \Phi_S(t)
\\
&=& 
\left(\Us(t)L_\ell\Us(t)\right)^*\otimes \left(\Us(t)L_\ell\Us(t)\right)
\\
&& -\frac{1}{2}\left(
I_S\otimes\left(\Us(t)^\dag(L_\ell^\dag L_\ell)\Us(t)
\right)^T \right. \\
&&\left. + \Us(t)^\dag(L_\ell^\dag L_\ell)\Us(t) \otimes I_S
\right)
\eea
\eeq
%

If there is sufficient control to make the time-averages of the
coherent error $\avg{\Hscoh}=0$ and the Lindblad terms
$\avg{G_\ell}=0,\forall\ell$, then the robust performance limit from
Theorem~\ref{avgThm} would correspond to the smallest intrinsic time-bandwidth
uncertainty error bound, \ie, $\tf\Omunc$.
Though the Lindblad form captures open-system behavior, the starting assumption is that the 
initial system-bath state is factorized. This is highly unlikely to be the case, but nevertheless, 
we can consider the Lindblad form to be a nominal model of the system through which a control can be designed.
If a control based on the Lindblad model produces a sufficiently small predicted time-bandwidth 
uncertainly level $\tf\Omeff$, then it is possible that unknown uncertainties can be withstood, 
including initial state coupling errors.

\subsection{Ancilla}

The link to error correction requires ancilla qubits, resulting in the
following modification of the bipartite system block diagram \cref{eq:SBsys}
to the \emph{tripartite} system 
\beq[eq:SABsys]
\bea{ll}  
\bea{rcl}
\bea{cc}
S,A & \vlongrightarrow{0.5in}
\\
& \ket{\psi(0)}
\\
B & \vlongrightarrow{0.5in}
\eea
&
\hspace{-2.3ex}
\mathbox{
  \bea{ccc}
  &&\\&U(t)&\\&&
  \eea
}
&
\hspace{-2.3ex}
\bea{cc}
\vlongrightarrow{0.5in} & {S,A}
\\
\ket{\psi(t)} & 
\\
\vlongrightarrow{0.5in} & {B}
\eea
\eea
\eea
  \eeq
There are now three types of states: $\ns$ system states, $\na$
ancilla states, and $\nb$ bath states with the total Hamiltonian,
\beq[eq:Ham sab]
H(t) = H_{SA}(t)\otimes\Ib
+\Is\otimes\Ia\otimes\Hb + H_{SAB}
\eeq
The uncertainty-free (nominal) system-ancilla (SA) Hamiltonian is, 
\beq[eq:Ham sa]
H^{\rm nom}_{SA}(t) = \Hs(t)\otimes\Ia + \Is\otimes\Ha(t)
+\sumalf S'_\alf\otimes A'_\alf
\eeq
with associated SA system coherent errors,
\beq[eq:Ham sa coh]
H^{\rm coh}_{SA}(t) = \Hscoh(t)\otimes\Ia + \Is\otimes\Hacoh(t)
+\sumalf \eps_\alf S'_\alf\otimes A'_\alf
\eeq
and where coupling of SA states to the bath is given by,
\beq[eq:Ham sb ab]
H_{SAB} = \sum_\bet S_\bet\otimes\Ia\otimes B_\bet
+ \sum_\gam \Is\otimes A_\gam\otimes B_\gam
\eeq

\subsection{Multilevel systems}

Extra levels that are excluded from the basic model are easily accounted for, \eg, 
a qutrit as the system and then an extra level that is excluded. 
The first step is to express the total system Hamiltonian as,
\beq[eq:Hmult]
\beasep{1,25}{rcl}
H(t) &=& \Hm(t)\otimes\Ib + I_{\cal M}\otimes\Hb + H_{\Mcal B}
\\
\Hm(t) &=& \mat{H_S(t)&H_{SE}\\H_{SE}^\dag&H_E},
\\
H_{\Mcal B} &=& \sumalf M_\alf\otimes\Balf
\eea
\eeq
Here $\Hm(t)$ is the multilevel Hamiltonian of dimension $\ns+\ne$
where $S$ denotes the $\ns$ system states which carry the information,
and $E$ denotes the $\ne$ extra (multi) levels, \eg, $\ne=1$ for a
qutrit when the system is a qubit. The bath is again denoted by $B$ with $\nb$ bath
states. The total system dimension is $n=(\ns+\ne)\nb$. 

To illustrate the modeling procedure, assume that $\Hs(t)$ is
uncertainty-free and with known time-variations due to the control
fields (coherent errors are easily added).  The remaining
Hamiltonians are assumed to be constant and uncertain.  Following
\cref{eq:Rint}, define the interaction-picture unitary $\Ui(t)$ with $\dot U(t)
= -iH(t)U(t)$ via,
\beq[eq:Rint mult]
\beasep{1.25}{l}
U(t) = \left(U_\Mcal(t) \otimes \Ub(t)\right)\Ui(t)
\\
U_\Mcal(t) =
\mat{U_S(t)&U_{SE}(t)\\U_{ES}(t)&{U}_E(t)}
\eea
\eeq
where $\dot U_{\cal M}(t)=-i\Hm(t)U_{\cal M}(t)$.
Under these conditions, the interaction-picture unitary evolution and
interaction-picture Hamiltonian are,
\beq[eq:Rint mult evo]
\beasep{1.25}{rcl}
\dot \Ui(t) &=& -i\Hi_{\Mcal B}(t)\Ui(t)
\\ 
\Hi_{\Mcal B} &=& \sumalf \Hi_\Mcal^\alf(t) \otimes \Hi_B^\alf(t)
\\
\Hi_\Mcal^\alf(t) &=& U_\Mcal(t)^\dag M_\alf U_\Mcal(t)
\\
\Hi_B^\alf(t) &=& \Ub(t)^\dag B_\alf \Ub(t)
\eea
\eeq
These interaction-picture Hamiltonians have the same form as in \cref{eq:Gint}.
To maximize fidelity to achieve a target $W_S$ in the system,
despite uncertainties, we ensure that $\Fnom=1$ ($\Us(\tf)=W_S$) and
simultaneously minimize the time-averaged terms involving the
controlled unitary $\Us(t)$ using reduced-order models of the
uncertain terms in the multilevel interaction-picture Hamiltonian as well as
the bath terms. With sufficient control resources, the time-bandwidth
uncertainty then only depends on the intrinsic (multilevel)
system-bath coupling bound,
\beq[eq:TOmunc mlev]
\beasep{1.25}{rl}
\tf\Omunc \geq 
& \sumalf \norm{M_\alf} \norm{\tf B_\alf}
\eea
\eeq
%
\subsection{Crosstalk}
Unwanted interactions can occur within the system, the latter being nullified (ideally) by
control; see, \eg,~\cite{tripathi2021suppression,Zeyuan:22,brown2024efficient}. Conventionally, the system is divided into ``main'' and ``spectator'' qubits, with the former performing the computation in a $\ns$-dimensional Hilbert space while the latter occupy a $\nq$-dimensional Hilbert space and represent the unwanted coupled states. In this case, the total dimension should include not only the
bath but \emph{all} the spectator states present during the operation time.  Thus the ``B''-system
dimension [\cref{eq:SBsys}] is not just the bath, but also the
crosstalk induced by these unwanted interactions, 
resulting in a total dimension $d=\ns \nq \nb$.  The
spectator qubits can be considered as part of the uncertain environment.

The Hamiltonian structure is similar to that of the multilevel system
\cref{eq:Hmult} where now $H_{\Xcal}(t)$ replaces $H_{\cal M}(t)$
resulting in,
\beq[eq:Hcross]
\beasep{1.25}{rcl}
H(t) &=& H_{\Xcal}(t)\otimes\Ib + I_{\Xcal}\otimes\Hb +H_{{\Xcal}B}
\\
H_{\Xcal}(t) &=& \Hs(t)\otimes I_Q + \Is\otimes H_Q(t)
\\
H_{{\Xcal}B} &=& \sumalf X_\alf \otimes \Balf
\eea
\eeq
Again following \cref{eq:Rint}, define the interaction-picture unitary $\Ui(t)$
with $\dot U(t) = -iH(t)U(t)$ via,
\beq[eq:Rint cross]
\beasep{1.25}{l}
U(t) = \left(U_\Xcal(t) \otimes \Ub(t)\right)\Ui(t)
\\
U_\Xcal(t) = \Us(t)\otimes U_Q(t)
\eea
\eeq
Clearly, the robustness limit bound still applies with a
redefinition of the minimum possible time-bandwidth uncertainty bound,
\ie, the horizontal axis in \cref{lim}.  Specifically, if the
nominal fidelity $\Fnom=1$, then $\Us(\tf)=\phi_S\Ws,|\phi_S|=1$,
$U_Q(\tf)=\phi_QI_Q,|\Phi_Q|=1$. As a result
$U_\Xcal(\tf)=\phi_S\phi_Q\Ws\otimes I_Q$.  The minimum possible
time-bandwidth uncertainty bound is then,
\beq[eq:TOmunc cross]
\beasep{1.25}{rl}
\tf\Omunc \geq 
& \sumalf \norm{X_\alf} \norm{\tf B_\alf}
\eea
\eeq
\section{Bound for a general $W_B$}
\label{sec:Fnuc}

\begin{figure}[h]
    \centering
    \includegraphics[width=3.5in]{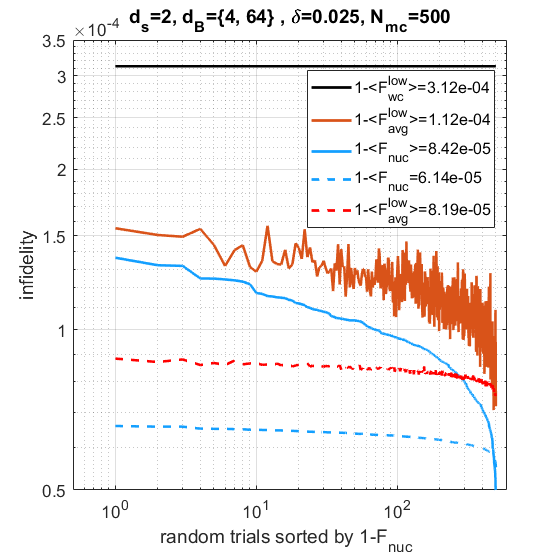}
    \caption{The plots compute infidelity bounds comparing the limit $1-\Flb$ from \cref{eq:Fbnd}, $1-\Fwc$ from \cref{eq:Fwc}, and 
$1-\Fnuc$ from \cref{eq:PhiB}. The bounds shown are for $\Ui=\exp{i\del H}$ with 
$\del=0.025$ for $N_{mc}=500$ random normalized 
$H,\normtwo{H}=1$ for two different bath dimensions, $\nb=\{4,64\}$ with {\bf black} for the limit bound $1-\Flb$, red for $1-\Favelow$, blue for $1-\Fnuc$ 
with solid lines for $\nb=4$ and dashed lines for $\nb=64$.}
    \label{nuc}
\end{figure}

Instead of comparing the final-time unitary to $\Ws\otimes\Ub(\tf)$, replace the final-bath unitary 
$\Ub(\tf)$ with the $\nb\times\nb $ unitary $W_B$, a free variable. Now 
define the error as $\norm{U(\tf)-\Ws\otimes W_B}$.
Using the final-time interaction transformation
$U(\tf)=\left( \Us(\tf)\otimes\Ub(\tf)\right)\Ui(T)$ together with $\Fnom=1$, \ie, 
$\Us(\tf)=\phi\Ws,|\phi|=1$, gives the error as,
\beq[eq:err Phib]
\begin{aligned}
\norm{U(\tf)-\Ws\otimes W_B} &= \norm{\Ui(\tf)-\Is\otimes\Phi_B}\\
\Phi_B &= \phi^*\Ub(T)^\dag W_B
\end{aligned}
\eeq
Following \cite{GDKBH:10}, for any $d\times d$ final-time interaction unitary
where $\Ui \equiv \Ui(\tf)$ with $d=\ns\nb$,
\beq[eq:PhiB]
\begin{aligned}
&\ {\ds\min_{\PhiB}} \normfro{\Ui-\Is\otimes\PhiB}^2 = 2d\left( 1 - \Fnuc \right)\\
\Fnuc &= 1 - (1/2d) {\ds\min_{\PhiB}} \normfro{\Ui-\Is\otimes\PhiB}^2 = \nucnorm{\Gam/d}\\
& \equiv (1/d)\sum_{i=1}^\nb {\rm sv}_i(\Gam), \quad 
\Gam = \sum_{i=1}^\ns \Ui_{[ii]}
\end{aligned}
\eeq
where $\Ui_{[ii]}$ are $\nb\times\nb$ submatrices of $\Ui$ along the block diagonal, and ${\rm sv}_i(\Gam)$ denotes the singular values of $\Gam$.
The minimizer $\PhiB^{\rm opt}$ is obtained from the SVD of $\Gam$, 
\beq[eq:Gam2PhiB]
\Gam = V_L 
\diag\left[{\rm sv}_1(\Gam),\cdots{\rm sv}_\nb(\Gam)\right] 
V_R^\dag
\quad\Rightarrow\quad
\PhiB^{\rm opt} = V_LV_R^\dag
\eeq 
Since $\normfro{\Ui-\Is\otimes\PhiB^{\rm opt}}\leq \normfro{\Ui-I}$,
it follows that,
\beq[eq:Fnuc comp]
\Fnuc \ge \Favelow \ge \Fwclow
\eeq
\cref{nuc} shows two numerical examples showing 
the limit infidelity bound $1-\Flb$ [$\Flb$ 
from \cref{eq:Fbnd}] which bounds $1-\Fwc$ [$\Fwc$ from \cref{eq:Fwc}]
also over-bounds $1-\Fnuc$. 


\bibliography{bib}  

\end{document}